\documentclass[11pt, letterpaper]{article}

\newif\ifdraft
\draftfalse
\ifdraft
\usepackage{refcheck}
\fi

\usepackage{amsthm, amsmath, amssymb}
\usepackage{mathtools} 
\usepackage[utf8]{inputenc}
\usepackage[normalem]{ulem}
\usepackage{dsfont}
\usepackage{enumitem} 
\usepackage{authblk}
\usepackage[subtle]{savetrees}
\usepackage[margin=1in,letterpaper]{geometry}

\usepackage{fullpage}
\usepackage[usenames,dvipsnames,svgnames]{xcolor}
\usepackage{xspace}
\usepackage{hyperref}
\usepackage{caption} 
\usepackage{subcaption}
\usepackage[section]{algorithm}
\usepackage[noend]{algorithmic}
\usepackage{graphicx}
\usepackage{subcaption}
\usepackage{tikz}
\usepackage{calc}

\definecolor{grey}{rgb}{0.8, 0.8, 0.8}
\definecolor{darkgrey}{rgb}{0.5, 0.5, 0.5}

\renewcommand{\hat}{\widehat}
\renewcommand{\tilde}{\widetilde}

\newcommand{\BT}{\begin{theorem}} \newcommand{\ET}{\end{theorem}}
\newcommand{\BPF}{\begin{proof}} \newcommand {\EPF}{\end{proof}}
\newcommand{\BPFOF}{\smallskip \begin{proof}} \newcommand {\EPFOF}{\end{proof}}
\newcommand{\BEQ}{\begin{equation}} \newcommand{\EEQ}{\end{equation}}
\newcommand{\BEQN}{\begin{eqnarray}}\newcommand{\EEQN}{\end{eqnarray}}
\newcommand{\BL}{\begin{lemma}} \newcommand{\EL}{\end{lemma}}
\newcommand{\BCM}{\begin{claim}} \newcommand{\ECM}{\end{claim}}
\newcommand{\BF}{\begin{fact}} \newcommand{\EF}{\end{fact}}

\newcommand{\algorithmicbreak}{\textbf{break}}
\newcommand{\BREAK}{\STATE \algorithmicbreak}

\DeclarePairedDelimiter{\ceil}{\lceil}{\rceil}

\newcommand{\norm}[1]{\left\lVert#1\right\rVert}
\newcommand{\tvdist}[1]{\left\lVert#1\right\rVert_{\mathsf{TV}}}

\newcommand{\naturals}{\mathbb{N}}

\newcommand{\vd}{\textbf{d}}

\newcommand{\etal}{{\it et~al.}\xspace}
\newcommand{\ie}{{\it i.e.,}\xspace}
\newcommand{\eg}{{\it e.g.,}\xspace}

\newcommand{\etc}{{\it etc.}\xspace}

\def\indicator{\mathds{1}}

\def\E{\mathbb{E}}

\renewcommand{\Pr}{\mathbb{P}}

\newcommand{\twonorm}[1]{\ensuremath{\norm{{#1}}_2}\xspace}

\newcommand{\tmix}{\ensuremath{t_{\rm mix}}\xspace}

\newlength\myindent

\usepackage{setspace}   

\def\th{^{th}}

\hypersetup{colorlinks=true, linktocpage=true, pdfstartpage=1, pdfstartview=FitV,%
    breaklinks=true, pdfpagemode=UseNone, pageanchor=true, pdfpagemode=UseOutlines,%
    plainpages=false, bookmarksnumbered, bookmarksopen=true, bookmarksopenlevel=1,%
    hypertexnames=true, pdfhighlight=/O,
    urlcolor=darkred, linkcolor=lightblue, citecolor=darkgreen, 
    pdfauthor={},%
    pdfsubject={},%
    pdfkeywords={},%
    pdfcreator={pdfLaTeX},%
    pdfproducer={LaTeX}%
}

\definecolor{darkred}{rgb}{0.5,0,0}
\definecolor{lightblue}{rgb}{0,0.4,0.8}
\definecolor{darkgreen}{rgb}{0,0.5,0}

\usepackage{aliascnt}
\def\NewTheorem#1#2{%
  \newaliascnt{#1}{theorem}
  \newtheorem{#1}[#1]{#2}
  \aliascntresetthe{#1}
  \expandafter\def\csname #1autorefname\endcsname{#2}
}

 \newtheorem{theorem}{Theorem}[section]
\NewTheorem{lemma}{Lemma}
\NewTheorem{corollary}{Corollary}
\NewTheorem{conjecture}{Conjecture}
\NewTheorem{observation}{Observation}
\NewTheorem{definition}{Definition}
\NewTheorem{proposition}{Proposition}
\NewTheorem{remark}{Remark}

\newtheorem{claim}[theorem]{Claim}
\theoremstyle{remark}

\newcommand\fund[1]{\ifx&#1& \else $#1$-\fi permanent\xspace}

\newcommand{\davg}{d_{\rm avg} }

\newcommand{\dvec}{{\bf d}}

\newcommand{\initquery}{\mathsf{init}}
\newcommand{\etype}{\mathsf{etype}}
\newcommand{\ein}{\mathsf{in}}
\newcommand{\eout}{\mathsf{out}}
\newcommand{\uOracle}{\mathcal{O}}
\newcommand{\dOracle}{\protect\overrightarrow{\mathcal{O}}}
\newcommand{\dOracleOne}{\protect\overrightarrow{\mathcal{O}}(1)}
\newcommand{\dOracleTwo}{\protect\overrightarrow{\mathcal{O}}(2)}
\newcommand{\rnull}{\mathsf{null}}



\renewcommand{\Pr}[1]{\mathbb{P}\left[\,#1\,\right]}
\renewcommand\E[1]{\mathbb{E}\left[\,#1\,\right]}

\renewcommand{\epsilon}{\varepsilon}

\title{ How large is your graph? 
}

\author[1,2]{Varun Kanade\thanks{varunk@cs.ox.ac.uk}}
\author[3,4]{Frederik Mallmann-Trenn\thanks{frederik.mallmann-trenn@ens.fr
}}
\author[3,5]{Victor Verdugo\thanks{vverdugo@dii.uchile.cl}}
\affil[1]{Department of Computer Science, University of Oxford}
\affil[2]{The Alan Turing Institute}
\affil[3]{D\'{e}partment d'Informatique, \'Ecole normale sup\'erieure}
\affil[4]{Department of Computer Science, Simon Fraser University}
\affil[5]{Departmento de Ingenier\'ia Industrial, Universidad de Chile}

\begin{document}

\begin{titlepage}
\maketitle
\thispagestyle{empty}

\begin{abstract}
	We consider the problem of estimating the graph size, where one is given
	only local access to the graph.  We formally define a query model in which
	one starts with a \emph{seed} node and is allowed to make queries about
	neighbours of nodes that have already been seen.  In the case of undirected
	graphs, an estimator of Katzir~\etal (2014) based on a sample from the
	stationary distribution $\pi$ uses $O\left(\frac{1}{\twonorm{\pi}} +
	\davg\right)$ queries; we prove that this is tight. In addition, we
	establish this as a lower bound even when the algorithm is allowed to crawl
	the graph arbitrarily; the results of Katzir~\etal give an upper bound that
	is worse by a multiplicative factor $\tmix \cdot \log(n)$.
	
	The picture becomes significantly different in the case of directed graphs.
	We show that without  strong assumptions on the graph structure, the number
	of nodes cannot be predicted to within a constant multiplicative factor
	without using a number of queries that are at least linear in the number of
	nodes; in particular, rapid mixing and small diameter, properties that most
	real-world networks exhibit, do not suffice. The question of interest is
	whether any algorithm can beat breadth-first search. We introduce a new
	parameter, generalising the well-studied conductance, such that if a suitable
	bound on it exists and is known to the algorithm, the number of queries
	required is sublinear in the number of edges; we show that this is tight.
\end{abstract}

\end{titlepage}

\pagenumbering{arabic}

\section{Introduction}

Networks contain a wealth of information and studying properties of networks
may yield important insights. However, most networks of interest are very large
and ordinary users may have rather restricted access to them. One of most basic
questions about networks is the number of nodes contained in them. For example,
the number of pages on the world wide web (WWW) is estimated to be just shy of
50 billion at the time of writing.\footnote{\url{www.worldwidewebsize.com}}
Facebook currently reports having about one and three quarter billion
users;\footnote{Here, billion refers to one thousand million, not a million
million.} Twitter reports having about 300 million active users. It is
undesirable to rely on a small number of sources for such information. At times
we might be interested in more specific graphs for which there is no public
information available at all. Is there a way to estimate the total number of
nodes using rather limited access to these graphs?

Our model is motivated by the kind of access ordinary users may have to graphs of
interest. For example, most social network companies provide some sort of an
application programming interface (API). In the case of the world wide web, one
option is to simply crawl. The graph query access models we use are formally
defined in \autoref{sec:querymodel}. For now, we mention four specific
networks, each of which captures an important modelling aspect. First, Facebook
is an undirected graph of friendships, and as long as privacy settings allow
it, it is possible to request the number of friends for a given user and the
identity of the friends. Second, the world wide web, which is a directed
graph---it is possible to extract out-links on a given webpage; however, there
is no obvious method to access all in-links. The third is what we refer to as
the ``fan'' network---for a specific user it is possible to query who she is a
fan of, however in terms of her fans only the number is
revealed.\footnote{Although it is not obvious which of the extant networks has
this property, there are close approximations such as Blogger. In any case, it
is a very natural intermediate model between the second and the fourth.} And
finally, the fourth is the twitter network, which is directed, however both the
followers and followees of a given user are accessible, as far as privacy
settings allow. Obviously, the method to estimate the number of nodes may be
rather different in each case. It is worth mentioning that the twitter network
can essentially be treated as an undirected graph, however it still leaves open
the possibility that differentiating in-links and out-links leads to better
estimators. 

While memory and computational requirements do act as constraints when dealing
with large graphs, possibly the most important one is the rate limits set on
queries made using the API. Even when crawling the web, there is the risk of
simply being blocked if large volumes of requests are sent to the same server.
Thus, the most scarce resource in this instance is the number of queries made
about the graph. Our query model counts these costs strictly---essentially every
time a new node is discovered, its degree is revealed at unit cost and a unit
cost is incurred for every neighbour requested. 

\subsection{Related Work}

There is a large body of literature on estimating statistical properties of
graphs, reflecting the relevance of and interest in studying complex networks.
Much of this work is in the more applied literature and in particular, we were
unable to pin down a precisely defined query model. The work most closely
related to ours is that of Katzir~\etal~\cite{KLSC14}. They consider a model
where a random sample drawn from the stationary distribution of the random walk
is available. If the random walk  mixes rapidly and a suitable bound on the
mixing time is known, this can be simulated in the models we consider. They
show that $O\left(\frac{1}{\twonorm{\pi}} + \davg\right)$ queries suffice,
where $\pi$ is the stationary distribution and $\davg$ is the average degree.
We show that this is tight when given access to a sample from the stationary
distribution. When only neighbour queries are allowed, we show that this bound
is tight up to a multiplicative factor of the mixing time and other
polylogarithmic factors. It is worth mentioning what this bound actually yields
in graphs where the average degree is small, something common to most real
world networks. It is always the case that $\twonorm{\pi}^{-1} \leq \sqrt{n}$,
where equality holds in the case of regular graphs. Thus, the number of queries
required is significantly sublinear in the number of nodes.  For graphs with
power law degree distributions with parameter $\beta = 2$, Katzir~\etal
calculated that $\twonorm{\pi}^{-1} = O(n^{1/4} \log n)$.
See \autoref{sec:them} for further discussion. 

In more recent work, Hardiman and Katzir give another estimator based on
counting shared neighbours~\cite{HK:13}. They give a slightly better bound on
the number of nodes sampled than the earlier work of
Katzir~\etal~\cite{KLSC14}.  However, it is unclear that this can be
implemented efficiently in the query model in our paper.
Cooper~\etal~\cite{CRS12} obtain estimators for the number of edges, triangles
and nodes. The estimators are based on random walks over the graph, but in
particular, for estimating the number of nodes, the transition probabilities
are not inversely proportional to the degree. Their estimator relies on
counting the time required for the $k\th$ return to a particular vertex. It
seems unlikely that either their random walk or their estimator can be
implemented in a query efficient manner.  Musco~\etal~\cite{MSL16} develop a
distributed collision based approach where several random walks traverse the
graph and counts the number of collisions with other random walks.
Dasgupta~\etal~\cite{DKS14} provide an estimator of the average degree that
uses $O(\log U\log \log U)$ samples, where $U$ is an upper bound on the maximum
degree. Somewhat surprisingly, this seems to be an easier task than estimating
the number of nodes.

Cooper and Frieze~\cite{CF02} estimate the graph size in a dynamic setting
where the graph grows over time and there exists an agent that is allowed to
move through the network every time a new node arrives. They prove that this
agent visits asymptotically a constant fraction of the vertices. If one has
access to the complete neighbourhood of a node when it is visited, it is shown
by Mihail~\etal that the expected time in which a walk discovers a power law
random graph is sublinear~\cite{MST06}. There has been some older theoretical
work on estimating the size of graphs motivated primarily by search
algorithms~\cite{knuth74, marchetti88}. However the model of graph access is
unrelated and they do not present bounds on the estimates. 

Related to the question of estimating graph properties is that of testing
whether graphs have specific properties. Property testing has received much
attention in the last two decades with properties of undirected graphs being
one of the important areas of focus (see \eg~\cite{Gol:2010}). Directed graphs
have received somewhat less attention; there are two main query models,
\emph{unidirectional} where only out-neighbours may be queried, and
\emph{bidirectional} where both in and out-neighbours may be
queried~\cite{BR:2002}. Bender and Ron showed that there exist properties such
as strong connectivity, where the query complexity may be either $O(1)$ or
$\Omega(\sqrt{n})$ depending on the model used, even when allowing two-sided
error. More recently, Czumaj~\etal have shown that if a property can be tested
with constant query complexity in the bidirectional setting, it can be tested
with sublinear query complexity in the unidirectional model~\cite{CPS:2016}.
Although, these query models are closely related to the ones in this paper, a
crucial aspect exploited by most property testing algorithms is the ability to
sample nodes uniformly at random, something that is not available in our
setting.

A closely related line of work is that on estimating properties of
distributions from samples; of most interest, in our case is the support size.
This problem has a long history going back at least to Good and
Turing~\cite{Goo:1953}. We only mention a few recent relevant results here.
Given access to uniformly sampled elements from a set, $O(\sqrt{n})$ samples
suffice to derive a good estimate of the set size using the birthday
problem~\cite{Finkel98}.  However, it is not clear how to sample from the
uniform distribution over the nodes of the graph in the query model we
consider. Valiant and Valiant show that support size can be estimated to a good
accuracy using $O(n /\log n)$ samples for any distribution~\cite{VV11, VV13}.
However, their result requires that any element in the support has
$\Omega(1/n)$ probability mass, something that is not true for the stationary
distribution of a random walk on a graph. (In the case of directed graphs this
problem becomes even more severe, since some nodes may have exponentially small
stationary probability mass.)

\subsection{Our contributions}

Firstly, our contribution is to express the problem of estimating graph
properties formally. As discussed in the introduction, networks of interest
vary significantly in terms of what access might be easily available to an
ordinary user. Keeping in mind the examples of Facebook, the web, the
fan-network and Twitter, we introduce different types of oracles that provide
access to the graph. 

The focus of this work is on estimating the number of nodes. For any $\epsilon
> 0$ and $\delta > 0$, we say that an algorithm (with access to a query oracle)
provides an $\epsilon$-accurate estimate of the number of nodes, if with
probability at least $1 - \delta$ it outputs $\hat{n}$, such that $|n -
\hat{n}| \leq \epsilon n$. The main quantity to be optimised is the number of
oracle queries, though all algorithms considered in this paper are also
computationally highly efficient. Allowed queries are defined precisely in
\autoref{sec:querymodel}. Here, we point out that a unit cost must be paid for
every disclosed neighbour as well as to know the degree of a node. All
algorithms have access to the \emph{identifier} of one \emph{seed} node in the
graph to begin with. Throughout we will assume that $\epsilon$ and $\delta$ are
constants and the use of $O(\cdot)$ and $\Omega(\cdot)$ notation in this paper
hides all dependence on $\epsilon$ and $\delta$. \medskip

\noindent {\bf Undirected Graphs}.
Katzir~\etal~\cite{KLSC14} implicitly assume the ability to sample from the
stationary distribution. They show that in this setting
$O\left(\frac{1}{\twonorm{\pi}} + \davg\right)$ samples from the stationary
distribution $\pi$ suffice, where $\davg$ is the average degree. If the graph
is connected and a suitable bound on the mixing time exists which is known to
the algorithm, $O(\tmix \log n)$ queries suffice to draw one node from a
distribution that is close (up to inverse polynomial factors in variation distance) to the stationary distribution using only neighbour
queries. This gives an upper bound of $O\left(\tmix \cdot \log n \cdot \left(
\frac{1}{\twonorm{\pi}} + \davg\right)\right)$ queries with a neighbour query
oracle (\autoref{cor:from-klsc}). In terms of lower bounds, we establish
that 

\begin{itemize}
	\item (\autoref{thm:lower-bound-undirected}) Any algorithm that has
		access to random samples from the \emph{stationary distribution} $\pi$
		and outputs a $0.1$-accurate estimate of the number of nodes with
		probability at least $0.99$, requires $\Omega\left( \frac{1}{\twonorm{\pi}} + \davg
		\right)$ samples.
	\item (\autoref{thm:lboqmodel}) Any algorithm that has access to
		neighbour queries and outputs a $0.1$-accurate estimate of the number of
		nodes with probability at least $0.99$, requires $\Omega\left(\frac{1}{\twonorm{\pi}} + \davg \right)$ queries.
\end{itemize}

We remark that there is a gap between the upper bound and lower bound when
considering an oracle with neighbour access. A question left open by our work
is whether the multiplicative factor of $\tmix \log n$ is required, or whether
a more efficient estimator can be designed. \medskip 

\noindent {\bf Directed Graphs}. The estimator of Katzir~\etal is not
applicable in the setting when graphs are directed, unless the query model
allows in-neighbour queries as well as out-neighbour queries, in which case all
results in the undirected setting hold. The reason for this is that even if one
did receive a sample drawn from the stationary distribution, it is no longer
the case that $\pi_v \propto \deg(v)$, a crucial property exploited by the
estimator of Katzir~\etal. In fact, unless strong assumptions are made on the
relationship between the in-degree and the out-degree (\eg being Eulerian), no
simple expression for $\pi$ in terms of degrees exists.\footnote{We mention that
the distribution is closely related to the PageRank distribution. However, the
PageRank random walk jumps to a uniformly random node in the graph with a small
probability; this is done to avoid problems when the graph is not strongly
connected.}

We provide constructions of graphs that demonstrate that low average degree,
rapid mixing and small diameter are not sufficient to design algorithms to
estimate the graph size with sublinear (in $n$) query complexity. In fact, we
show that when only given access to a sample from the stationary distribution,
a superpolynomially large sample may be required even for graphs with constant
average degree, logarithmic diameter and rapid mixing. The reason for this is
that in the case of directed graphs most of the stationary mass can be
concentrated on a very small number of nodes.

In order to understand the query complexity when neighbour-queries are
allowed, we define a new parameter called \emph{generalized conductance}.
Roughly speaking a graph has $\epsilon$-general conductance $\phi_\epsilon$, if
every set containing at most $(1 - \epsilon) n$ nodes has at least
$\phi_\epsilon$ fraction of directed edges going out. If a suitable bound on
the value of $\phi_\epsilon$ is known, we show that a simple edge-sampling
algorithm outputs an estimate to within relative error $\epsilon$ while using
$O(n/\phi_{\epsilon})$ queries (\autoref{thm:upper}) and that this is
almost tight, in the sense that any algorithm that outputs an estimate that is
even slightly better than $\epsilon$ requires at least
$\Omega(n/\phi_{\epsilon})$ queries (\autoref{thm:philower}). The algorithm
only requires access to out-neighbours, while the lower bound holds even with
respect to an oracle that allows in-neighbour queries, which means that it also
applies in the case of undirected graphs.

\subsection{Discussion}

In terms of improvements to our results, the most interesting question is
whether any reasonable subclass of directed graphs are amenable to
significantly improved query complexity (ideally sublinear) for the problem of
estimating the number of nodes. The constructions in \autoref{sec:comet}
show that having low average degree, small diameter and rapid mixing is not
enough.  For undirected graphs, it is an interesting question whether the extra
factor of mixing time $\tmix$ must be paid, when only neighbour queries are
allowed.  It is conceivable that an improved estimator that can handle
correlated pairs of nodes can be designed, so as to not waste all but one query
for every $\tmix$ queries.  Finally, it'd be interesting to study the question
of estimating other properties of graphs, number of edges, number of triangles,
\etc in this framework.

The model choices we made reflect the publicly available access to most extant
networks; in particular, we were very stringent with accounting---every
neighbour query counts as unit cost.  Many APIs return the list of neighbours,
although in chunks of a fixed size, \eg 100. It is hard to argue that 100
should be treated as constant in the context of social networks.  Nevertheless,
if we wanted a list of all followers of Barack Obama, this would still result
in a very large number of queries. A natural extension to the query model in
this paper is to allow the entire neighbourhood (possibly restricted to
the out-neighbourhood in the case of directed graphs) to be  revealed at unit cost.
Estimators such as the one involving common neighbours of Hardiman and
Katzir~\cite{HK:13} can be implemented efficiently under such a model.
Understanding the query complexity of estimation in these more powerful models
is an interesting question.

\section{Model and Preliminaries}\label{sec:model}

{\it Graphs.} Graphs $G = (V, E)$ considered in this paper may be directed or
undirected; typically we assume $|V| = n$ and $|E| = m$, though if there is
scope for confusion we use $|V|$ or $|E|$ explicitly. For undirected graphs,
for a node $v \in V$, we denote by $N(v)$ its \emph{neighbourhood}, \ie $N(v) =
\{ w ~|~ \{v, w \} \in E \}$, and its degree by $\deg(v) := |N(v)|$. In the
case of directed graphs, we denote $N^+(v) := \{ w ~|~ (v, w) \in E \}$ its
\emph{out-neighbourhood} and by $\deg^+(v) := |N^+(v)|$ its \emph{out-degree}.
Similarly, $N^-(v) := \{ u ~|~ (u, v) \in E \}$ denotes its
\emph{in-neighbourhood} and $\deg^-(v) := |N^-(v)|$ its \emph{in-degree}.
Furthermore, $\davg$ denotes the average degree, i.e, $\sum_ {v\in
V}\deg(v)/n$.  Whenever there is scope for confusion, we use the notations
$\deg_G(u)$, $N_G(v)$, $\davg(G)$, \etc to emphasise that the terms are with
respect to graph $G$. \medskip

\noindent {\it Random walks in graphs.} A discrete-time lazy random walk
$(X_t)_{t\ge 0}$ on a graph $G=(V,E)$ is defined by a Markov chain with state
space $V$ and transition matrix $P=(p(u,v))_{u,v\in V}$ defined as follows: 
 For every $u \in V$, $p(u,u) = 1/2$ ({\it Laziness}).
In the undirected setting, for every $v \in N(u)$, $p(u,v) = 1/(2
		\deg(u))$. In the directed setting, for every $v \in N^+(u)$, $p(u,v) =
		1/(2 \deg^+(u))$. 
The transition probabilities can be expressed in matrix form as $P=(I+
\mathcal{D}^{-1}A)/2$, where $A$ is the adjacency matrix of $G$, $\mathcal{D}$
is the diagonal matrix of node degrees (only out-degrees if $G$ is directed),
and $I$ is the identity. Let $p^t(u, \cdot)$ denote the distribution over nodes
of a random walk at time step $t$ with $X_0 = u$. For the most part, we will
consider (strongly) connected graphs. Together with laziness, this ensures that
the stationary distribution of the random walk, denoted by $\pi$, is unique and
given by $\pi P=\pi$. In the undirected case, the form of the stationary
distribution is particularly simple, $\pi(u) = \deg(u)/(2|E|)$; furthermore, the
random walk is reversible, \ie $\pi(u)  p(u,v) = \pi(v) p(v,u)$. As before,
$\pi_G$ is used to emphasise that the stationary distribution is respect to
graph $G$. \medskip

\noindent {\it Mixing time.} To measure how far $p^t(u,\cdot)$ is from the
stationary distribution we consider the {\it total variation distance}; for
distributions $\mu,\nu$  over sample space $\Omega$ the total variation
distance is $\tvdist{\mu-\nu}=\frac{1}{2}\displaystyle\sum_{x\in \Omega}|\mu(x)-\nu(x)|$.  The
\emph{mixing time} of the random walk is defined as
$
	\tmix := \max_{u \in V} \min \{ t \geq 1 ~|~ \tvdist{p^t(u, \cdot) - \pi} \leq e^{-1} \}.
$
Although the choice of $e^{-1}$ is arbitrary, it is known that after
$\tmix\log(1/\varepsilon)$ steps, the total variation distance is at most
$\epsilon$.

\subsection{Query Model}\label{sec:querymodel}

In this section, we formally define the query model that allows us to access
the graph.  We consider four different neighbour query oracles, $\uOracle$,
$\dOracle$, $\dOracleOne$ and $\dOracleTwo$. We also assume that all oracles
have graphs stored as adjacency lists. In the case of directed graphs, there
are two adjacency lists for every vertex, one for in-neighbours and one for
out-neighbours. No assumption is made regarding the order in which the
adjacency lists are stored. 

All oracles also make use of labelling functions\footnote{We do assume that
there is unit cost in sending the label of a node to the algorithm, thus
implicitly we may think of every label having a bit-representation that is
logarithmic in the size of the graph. However, the bit-length of the label
reveals minimal information, which we assume that the algorithm has access to
anyway.} ; for some space $L$, a labelling function $\ell : V \rightarrow L$
used by an oracle is an injection.  We allow $\ell$ to be defined dynamically.
The labelling function we use throughout the paper is the \emph{consecutive
labelling function} defined as follows: The label set is $L = \naturals$. If
$S$ denotes all the vertices labelled by the oracle so far, for a new vertex $v
\not\in S$ picked by the oracle, the label is assigned as follows: if $S =
\emptyset$, $\ell(v) = 1$, else $\ell(v) = \max \{ \ell(u) ~|~ u \in S \} + 1$. 
For these neighbour query oracles, any algorithm can essentially make two types
of queries, (i) $\initquery$, and (ii) $(l, i)$ for undirected graphs or  $(l,
i, \etype)$ for directed graphs, where $\etype$ is either $\ein$, $\eout$. We
assume that oracles use a labelling function $\ell$.
\smallskip
\begin{flushleft}
	(i) $\initquery$: The oracle initialises a set $S := \{ v \}$, where $v$ is
	chosen to be an arbitrary node in the graph. The different oracles respond as
	defined below.
	\begin{itemize}
		\item $\uOracle$ responds with $(\ell(v), \deg(v))$ for some $v \in V$.
		 $\dOracle$ responds with $(\ell(v), \deg^+(v))$ for some $v \in V$.
		\item $\dOracleOne$ and $\dOracleTwo$ both respond with $(\ell(v), \deg^+(v),
			\deg^-(v))$.
	\end{itemize}
	(ii) $(l, i)$ or $(l, i, \etype)$: $\uOracle$ only responds to query of type $(l,
	i)$ and the remaining to queries of the type $(l, i, \etype)$.
	\begin{itemize}
		\item For query $(l, i)$, if there is $v \in S$ such that $\ell(v) = l$
			and $i \leq \deg(v)$, then $\uOracle$ returns $(\ell(u), \deg(u))$,
			where $u$ is the $i\th$ element in the adjacency list of $v$. The
			oracle updates $S \leftarrow S \cup \{ u \}$. Otherwise it returns
			$\rnull$. 
		\item For query $(l, i, \eout)$, if there is $v \in S$ such that $\ell(v)
			= l$ and $i \leq \deg^+(v)$, then $\dOracle$ returns $(\ell(u),
			\deg^+(u))$, while both $\dOracleOne$ and $\dOracle(2)$ return
			$(\ell(u), \deg^+(u), \deg^-(u))$ where $u$ is the $i\th$ element on
			the out-neighbour adjacency list of $v$. The oracles all update $S
			\leftarrow S \cup \{ u\}$. Otherwise, all oracles return $\rnull$. 
		\item For query $(l, i, \ein)$, $\dOracle$ and $\dOracleOne$ always
			return $\rnull$. If there exists $v \in S$, such that $\ell(v) = l$
			and $i \leq \deg^-(v)$, then $\dOracleTwo$ returns $(\ell(u),
			\deg^+(u), \deg^-(u))$, where $u$ is the $i\th$ element on the
			adjacency list of in-neighbours of $v$; otherwise, it returns
			$\rnull$. In the case of $\dOracleTwo$, it updates $S \leftarrow S
			\cup \{ u \}$.
	\end{itemize}
\end{flushleft}

In words, $\uOracle$ captures access to undirected graphs such as social
networks like Facebook, $\dOracle$ captures directed graphs such as the world
wide web, where only out-edges are available, $\dOracleOne$ captures directed
graphs such as fan-networks, where the in-degree but not in-neighbours may be
available, and $\dOracleTwo$ captures directed graphs such as Twitter where
access is available to both in-edges and out-edges.

It is worth pointing out that a response of $\rnull$ provides no new
information to the algorithm. The algorithm only knows that it hadn't received
$l$ as a label before or that the degree of the node with label $l$ is strictly
smaller than $i$, things it already knew.  This is because the oracle maintains
a history of past queries; if this were not the case the algorithm could
generate (random) labels and try to find out whether they corresponded to nodes
in the graph. However, even for oracles that don't maintain such state
explicitly, by choosing $\ell$ to be collision-resistant hash function,
essentially the same behaviour can be achieved.

In some of our proofs, we also allow oracles to return side information.  These
are denoted by a superscript $s$, \eg $\uOracle^s$. Access to oracles with
(truthful) side information can only help to reduce query complexity, since an
algorithm can choose to ignore any side information it receives. Finally, we
say that an algorithm is \emph{sensible} if it does not make a query to which
it already knows the answer. It is clear that given any algorithm, there is a
sensible algorithm that is at least as good as the original algorithm. The
\emph{sensible} algorithm merely simulates the original algorithm and whenever
the original algorithm made a query to which the answer was known, the
\emph{sensible} algorithm simulates the oracle response.

Finally, we consider the \emph{stationary query oracle}, $\uOracle^\pi$, which
when queried returns $(\ell(v), \deg(v))$, where $v \sim \pi$; $\pi$ is the
stationary distribution.

\section{Undirected graphs}\label{sec:undirected}

In this section, we focus on the question of estimating the number of nodes in
undirected graphs. We show that the results obtained by
Katzir~\etal~\cite{KLSC14} are essentially tight, up to a factor of the
mixing time and polylogarithmic terms in $n$. This suggests that rapid
mixing is a critical condition for being able to estimate the size of the
graph. We begin by discussing very simple examples which show that this is
indeed the case.

\paragraph{Some Simple Observations} 

First we construct two simple graphs, one having twice as many nodes as the
other. No algorithm can distinguish between these two graphs with probability
greater than $3/4$ (say) unless $\Omega(m)$ queries to the oracle $\uOracle$
are made. The graphs are shown in Figure~\ref{fig:mixing-required}. The first
graph is simply $G_{n, p}$ with each node being connected to an additional node
with degree $1$. The second graph is identical, except that for one randomly
chosen node, instead of this additional edge to a degree $1$ node, it is
connected to a copy of $G_{n, p}$. Suppose $p > \log(n)/n$ so that the graphs
are connected with high probability and that $m = \Theta (n^2p)$. 

We don't provide a formal proof here, which can be given along the lines of the
proof of \autoref{thm:lower-bound-undirected}. We outline a sketch here.
Assuming the oracle $\uOracle$ returns some $v \in V$ in the part that is
$G_{n, p}$ with an additional edge coming out of each node, how long will it
take before we discover the edge to the copy of $G_{n, p}$, if it exists? Since
we made no assumption on the order in which neighbours are stored in an
adjacency list, we have to query at least a constant fraction of neighbours of a
given node to discover, with constant probability, the edge that leads out of
the $G_{n, p}$. However, most of these will just lead us to a single vertex.
Thus, unless we explore at least a constant fraction of the edges of a constant
fraction of the vertices, we don't stand a significant chance of finding the
edge that leads to the copy of the $G_{n, p}$. Thus, $\Omega(m)$ queries to the
oracle are required. 

Clearly, if we're in the second case, that is the graph with an additional copy
of $G_{n, p}$ connected to the earlier one by a single edge, a random walk on
this graph does not mix rapidly. What this example shows is that if one wishes
to do asymptotically better than breadth-first search, one needs additional
assumptions on the graph, such as rapid mixing.

\begin{figure}
\centering
\begin{minipage}{.5\textwidth}
\begin{center}
\includegraphics[scale=0.4]{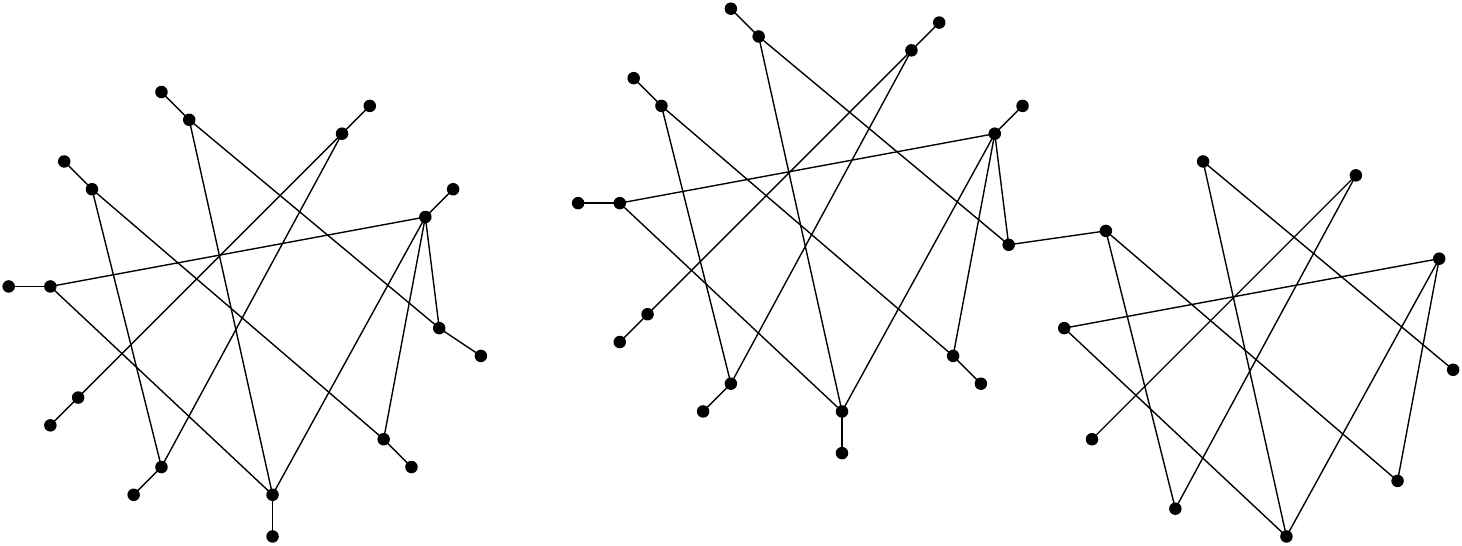}	
\end{center}
	\caption{Example showing that $\Omega(m)$ samples might be required if no assumption like rapid mixing is made. \label{fig:mixing-required}}
\end{minipage}%
\begin{minipage}{.5\textwidth}
\begin{center}
	\includegraphics[scale=0.4]{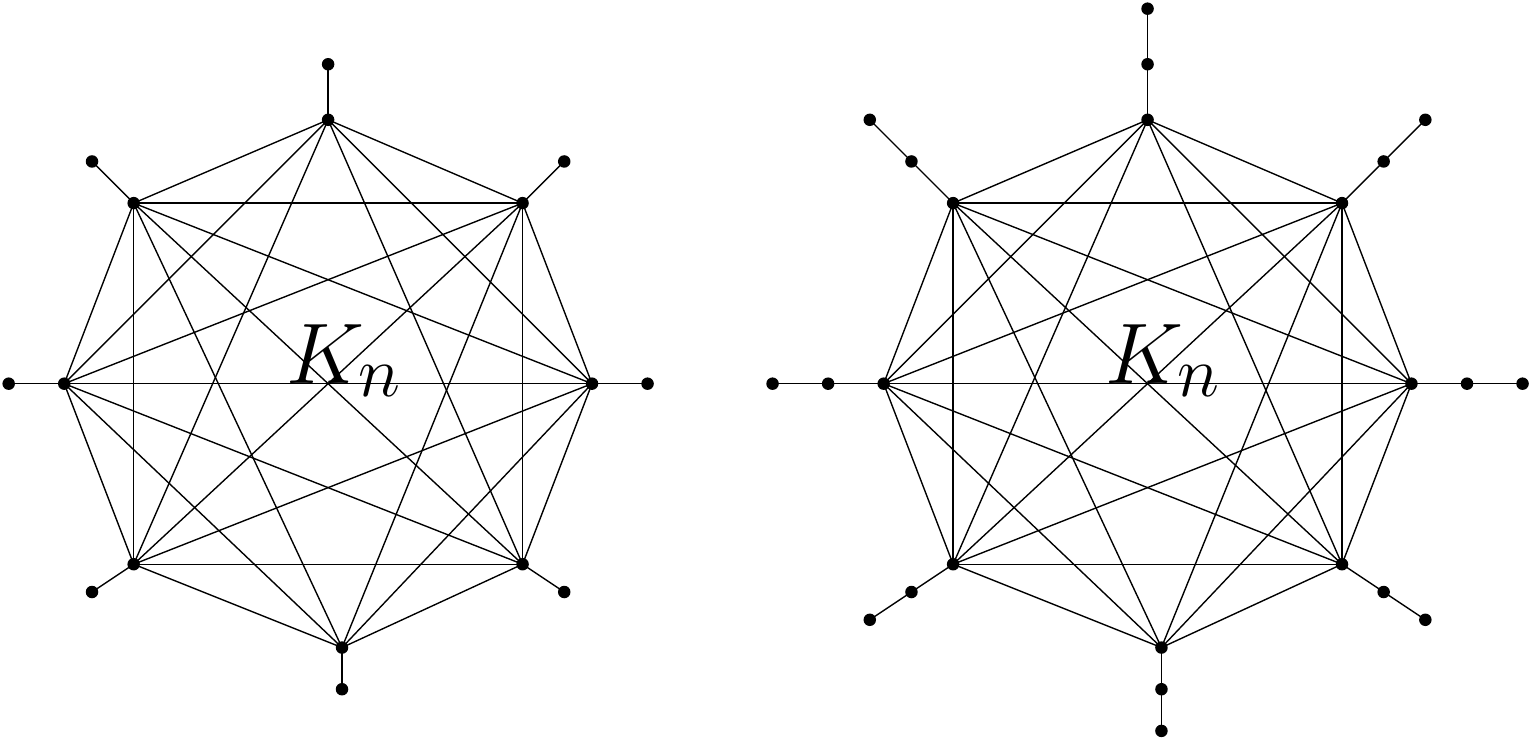}
\caption{ The \emph{sun} graph and the \emph{bright sun}.}
\label{fig:suns}
	
\end{center}
\end{minipage}
\end{figure}

\subsection{Results of Katzir~\etal}\label{sec:them}

In this section, we discuss the result of Katzir~\etal~\cite{KLSC14} regarding
estimating the number of nodes in a graph. Though, they don't discuss this
formally, their method essentially boils down to having access to the
stationary query oracle $\uOracle^\pi$. 

For the sake of completeness, we outline the estimator of Katzir~\etal here.
Let \[X = \{(\ell(x_1), \deg(x_1)), (\ell(x_2), \deg(x_2)), \allowbreak  \ldots, (\ell(x_r),
\deg(x_r)) \}\] be drawn from the stationary query oracle $\uOracle^\pi$.  Let
$\Psi_{1} = \sum_{i = 1}^r \deg(x_i)$ and $\Psi_{-1} = \sum_{i = 1}^r
1/\deg(x_i)$; let $C = \sum_{i \neq j} \indicator(\ell(x_i) = \ell(x_j))$ denote the random
variable counting the number of collisions. Then, it is fairly straightforward
to see that $\E{\Psi_1 \Psi_{-1}} = r + 2 n {r \choose 2} \twonorm{\pi}^2$ and
$\E{C} = 2 {r \choose 2} \twonorm{\pi}^2$. Katzir~\etal use the the following
as an estimator for the number of nodes: 
$
	\hat{n} := \frac{\Psi_{1} \Psi_{-1} - r}{C}
$
They prove the following result.

\begin{theorem}[Katzir et al. \cite{KLSC14}] \label{thm:klsc}
	Let $\varepsilon>0$ and $\delta>0$. Suppose that the number of samples is
	$r\ge 1+\frac{32}{\varepsilon^2 \delta} \max\left \{
		\frac{1}{\|\pi\|_2}, \davg\right\}$, where $\davg =
	2m/n$. Then, $\Pr{|\hat{n} -n|\ge \varepsilon n}\le \delta.$
\end{theorem}

Given a bound $T$ on $\tmix$ and access to oracle $\uOracle$, the stationary
query oracle $\uOracle^{\pi}$ can be approximately simulated. (We assume that
the graph is connected.) We simply perform a random walk on the graph for $T
\log(1/\rho)$ steps, if a sample from a distribution at most $\rho$ far from
$\pi$ in total variation distance is desired.  Using the above theorem, we get
the following straightforward corollary.  The proof of the corollary follows by
simulating $s$ random walks for  $O(T\log s)$ time steps (queries) ensuring
that each random walk has a total variation distance of $s^{-2}$ from $\pi$.
Using the above theorem, we get the following straightforward corollary. 

\begin{corollary}\label{cor:from-klsc} 
	Let $\varepsilon > 0$ and $\delta > 0$.  For a connected, undirected graph,
	$G = (V, E)$ let $T$ be such that $\tmix \leq T$. Then there exists an
	algorithm that given $T$ and access to oracle $\uOracle$, outputs $\hat{n}$,
	such that, 
	$ \Pr{|\hat{n} - n| \geq \varepsilon n} \leq \delta. $
	Furthermore, the number of queries made by the algorithm to $\uOracle$ is
	$O\left(Ts\log s \right)$, where
	$s=O\left(\frac{1}{\epsilon^2 \delta} \cdot \max\left\{ \frac{ 1}
	{\twonorm{\pi}}, \davg \right\}\right)$.
\end{corollary}

The key question is, are these bounds tight? In
\autoref{sec:stat-lower-bounds}, we give much stronger lower bounds, where
we show that for any valid degree sequence $\vd = (d_1, \ldots, d_n)$ on $n$
nodes, there exists a sequence $\vd^\prime$ on $2n$ nodes, such that any
algorithm with access to a stationary oracle for either sequence $\vd$ or
$\vd^\prime$, cannot distinguish between the two unless it makes
$\Omega\left(\frac{1}{\twonorm{\pi}}+ \davg \right)$ queries. Here, we discuss
simple instances where the bound in \autoref{thm:klsc} is tight; we don't
give formal proofs which are relatively straightforward.

\begin{itemize}
	\item Let $G_1$ and $G_2$ be $d$-regular graphs on $n$ and $2n$ nodes
		respectively, for $d < n$. Thus, samples from the stationary distribution
		are essentially just uniformly chosen random nodes in the graph. As the
		degrees are identical, they reveal no additional information. Thus, the
		algorithm has to query until a collision is observed, which requires
		$\Omega(\sqrt{n})$ queries. Clearly for these graphs
		$\frac{1}{\twonorm{\pi}} = \Theta(\sqrt{n})$.
	\item We define the \emph{sun} graph as a $K_n$ with an additional edge out
		of each vertex to a node with degree $1$. The \emph{bright sun} graph is
		the same as the sun graph, except that there is a path of length $2$
		coming out of each vertex, rather than just an edge (see Figure \ref{fig:suns}). Under the stationary
		oracle model, $\Omega(n)$ queries are required before any degree $1$ or
		$2$ nodes are returned since the total probability on these nodes is
		$O(1/n)$. As can be seen for these graphs $\davg = \Theta(n)$.

\end{itemize}

The above examples show that there are graphs for which the bound of
Katzir~\etal is tight. In this paper, we show a significantly stronger
statement---given any degree sequence $\vd$, there are graphs for which the
bound is almost tight.

\subsection{Lower Bounds in the Stationary Query Model}
\label{sec:stat-lower-bounds}

In this section, we show that for any undirected, connected graph $G = (V, E)$,
there exists an undirected, connected graph $\tilde{G}$, such that any
algorithm which has access to either the oracle $\uOracle^{\pi}(G)$ or
$\uOracle^{\pi}(\tilde{G})$, cannot distinguish between the two with
significant probability without making a large number of queries. Thus, it
cannot output an estimate $\hat{n}$ satisfying $|\hat{n} - |V|| < |V|/2$ with probability $2/3$.
\begin{lemma} \label{lem:perro}
	Let $G=(V, E)$ be an undirected, connected graph with $|V| = n$ and $|E|
	\geq n$. Then there exists a graph $\tilde{G}=(\tilde{V}, \tilde{E})$, with
	$|\tilde{V}| = 2n$, such that any algorithm given access to either
	$\uOracle^{\pi}(G)$ or $\uOracle^{\pi}(\tilde{G})$ with equal probability,
	cannot distinguish between the two with probability greater than
	$\frac{2}{3}$, unless it makes at least
	$\Omega\left(\frac{1}{\twonorm{\pi_G}}\right)$ queries. As a consequence, no
	algorithm can output $\hat{n}$ satisfying $|\hat{n} - n^*| < n^*/2$ w.p. at
	least $2/3$, where $n^* = n$ if the chosen graph is $G$ and $2n$ if it is
	$\tilde{G}$.
\end{lemma}

\begin{proof}
	We construct $\tilde{G} = (\tilde V, \tilde E)$ by taking $G =(V, E)$ and an
	identical copy, denoted $G^\prime = (V^\prime, E^\prime)$, and connecting
	them arbitrarily in such a way that the degrees of the nodes remain
	unchanged. We note that this can always be done. Since $|E| \geq n$, the graph $G$ is
	not a tree, and hence there is an edge $\{i, j\} \in E$ such that removing
	it does not disconnect $G$. Let $\{i^\prime, j^\prime\}$ be the
	corresponding edge in the copy of $G$, then removing $\{i, j\}$ and
	$\{i^\prime, j^\prime\}$ and adding the edges $\{i, j^\prime\}$ and
	$\{i^\prime, j\}$ gives $\tilde{G}$.

	For some $s \in \naturals$, let $(X_t)_{t=1}^s$ be $s$ independent samples
	drawn from $\pi_G$. Similarly, let $(\tilde{X}_t)_{t=1}^s$ be independent
	samples from $\pi_{\tilde{G}}$. We define a coupling between $(X_t)_{t=1}^s$
	and $(\tilde{X}_t)_{t=1}^s$. Suppose $X_t = i$, then we define $\tilde{X}_t
	= i$ with probability $\tfrac{1}{2}$ and $\tilde{X}_t = i^\prime$ (the copy
	of $i$ in $G^\prime$) with probability $\tfrac{1}{2}$. This ensures that
	$\tilde{X}_t$ is drawn from the distribution $\pi_{\tilde{G}}$. Note that the
	output of the oracle $\uOracle^\pi(G)$ (respectively
	$\uOracle^\pi(\tilde{G})$) is $(\ell(X_t), \deg_G(X_t))$ (respectively
	$(\ell(\tilde{X}_t), \deg_{\tilde{G}}(\tilde{X}_t))$).

	We assume the oracle simply uses a consecutive labelling scheme, \ie when
	$X_t$ is drawn from $\pi$, if there is no $t^\prime < t$, such that
	$X_{t^\prime} = X_t$, then $\ell(X_t) = \max \{ \ell(X_{t^\prime}) ~|~
	t^\prime < t \} + 1$. Note that this ensures that unless there exist $t \neq
	t^\prime$ and $t, t^\prime \leq s$, such that $X_{t} = X_{t^\prime}$, the outputs
	of oracles $\uOracle^\pi(G)$ and $\uOracle^\pi(\tilde{G})$ would be
	identical.

	Finally, define  
	\[ C_s = \sum_{t' < t \leq s} \indicator(X_t
	\neq \tilde{X}_{t^\prime}) \]
	The expectation of $C_s$ is $\E{C_s} = {s \choose 2} \twonorm{\pi_{G}}^2$. Then by Markov's inequality, for
	a suitably chosen constant $c$, if $s \leq \frac{c}{\twonorm{\pi_G}}$,
	$\Pr{C \geq 1} \leq 1/6$. Thus, with probability at least $\frac{5}{6}$, the
	outputs of oracles $\uOracle^{\pi}(G)$ and $\uOracle^{\pi}(\tilde{G})$ can
	be coupled perfectly.
\end{proof}

\begin{lemma} \label{lem:perra}
	Let $G=(V, E)$ be an undirected, connected graph with $|V| = n$ and $|E|
	\geq n$. Then there exists a graph $\tilde{G}=(\tilde{V}, \tilde{E})$, with
	$|\tilde{V}| = 2n$, such that any algorithm given access to either
	$\uOracle^{\pi}(G)$ or $\uOracle^{\pi}(\tilde{G})$ with equal probability,
	cannot distinguish between the two with probability greater than
	$\frac{2}{3}$, unless it makes at least $\Omega\left(\davg\right(G))$
	queries. As a consequence, no algorithm can output $\hat{n}$ satisfying
	$|\hat{n} - n^*| < n^*/2$ w.p.  at least $2/3$, where $n^* = n$ if the
	oracle chosen corresponds to $G$ and $n^* = 2n$ if it corresponds to
	$\tilde{G}$.
\end{lemma}
\begin{proof}
	We construct $\tilde{G} = (\tilde V, \tilde E)$ by taking $G =(V, E)$ and a
	$3$-regular expander on $n$ nodes, denoted $G^\prime = (V^\prime,
	E^\prime)$, and connecting them arbitrarily in such a way that the degrees
	of the nodes remain unchanged. (If $n$ is odd, we can use set $|V^\prime| =
	|V| + 1$.) We note that this can always be done. Since $|E| \geq n$, the
	graph $G$ is not a tree, and hence there is an edge $\{i, j\} \in E$ such
	that removing it does not disconnect $G$. Let $\{i^\prime, j^\prime\}$ be
	the corresponding edge in the copy of $G$, then removing $\{i, j\}$ and
	$\{i^\prime, j^\prime\}$ and adding the edges $\{i, j^\prime\}$ and
	$\{i^\prime, j\}$ gives $\tilde{G}$.

	For $s \in \naturals$, let $(X_t)_{t = 1}^s$ be $s$ independent samples
	drawn from $\pi_G$; similarly $(\tilde{X}_t)_{t = 1}^s$ a sample drawn from
	$\pi_{\tilde{G}}$. We define a coupling between $(X_t)_{t =1}^s$ and
	$(\tilde{X}_t)_{t=1}^s$ as follows. Let $p = \pi_{\tilde{G}}(V^\prime)$, the
	probability that a node sampled according to $\pi_{\tilde{G}}$ is one of the
	newly created ones. Let $Z_t  = 1$ with probability $p$ and $0$ otherwise.
	Then we have $\tilde{X}_t = X_t$ if $Z_t = 0$, otherwise $\tilde{X}_t$ is chosen
	uniformly at random from $V^\prime$, the newly added nodes. To see that this
	is a valid coupling, observe that for any $u \sim \pi_{\tilde{G}}$,
	conditioned on $u \not \in V^\prime$, $u$ is  distributed according
	to $\pi_G$.

	Note that the oracle $\uOracle^{\pi}(G)$ (respectively
	$\uOracle^{\pi}(\tilde{G})$) returns $(\ell(X_t), \deg_G(X_t))$ (respectively
	$(\ell(\tilde{X}_t), \deg_{\tilde{G}}(X_t))$). We assume that the oracles use
	a consecutive labelling function, \ie every new node is given the smallest
	as yet unused natural number.

	Let $C = \sum_{t = 1}^s Z_t$, then $\E{C} = ps$. Note that $p = 3
	|V^\prime|/(2|E| + 3|V|^\prime) \leq 3|V|/(2|E|) = 3/\davg(G)$. Thus, for a
	suitably chosen constant $c$, if $s \leq c \cdot \davg(G)$, by Markov's inequality
	$\Pr{C \geq 1} \leq 1/6$. Note that conditioned on $C = 0$, the outputs of
	the oracles $\uOracle^\pi(G)$ and $\uOracle^\pi(\tilde{G})$ can be coupled
	perfectly. Hence, no algorithm can distinguish between $G$ and $\tilde{G}$
	using fewer than $c\cdot  \davg(G)$ queries with probability greater than or equal
	to $2/3$.											\end{proof}

As a consequence of \autoref{lem:perro} and \autoref{lem:perra}, we get the
following theorem.

\begin{theorem}\label{thm:lower-bound-undirected}
	Given an undirected, connected graph $G=(V, E)$, there exist graphs
	$\tilde{G}$, $\bar{G}$ with $2|V|$ nodes and a constant $p < 1$, such that
	any algorithm that is given access to one of three oracles
	$\uOracle^\pi(G)$, $\uOracle^\pi(\tilde{G})$ and $\uOracle^\pi(\bar{G})$,
	chosen with equal probability, requires
	$\Omega\left(\frac{1}{\twonorm{\pi_G}} + \davg (G)\right)$ queries to
	distinguish between them with probability at least $p$. As a consequence,
	any algorithm that outputs $\hat{n}$, such that $|\hat{n} - n^* | < n^*/2 $
	requires at least $\Omega\left( \frac{1}{\twonorm{\pi_G}} + \davg(G)
	\right)$ queries, where $n^* = n$ if the graph is $G$ and $n^* = 2n$ if the
	graph is either $\tilde{G}$ or $\bar{G}$.
\end{theorem}

\subsection{Oracle Sampling from the Neighbour Query Model}

In this section, we show that with access to the oracle $\uOracle(G)$, any
algorithm that predicts the number of nodes in a graph $G$ to within a small
constant fraction requires $\Omega\left(\frac{1}{\|\pi_G\|_2} + \davg \right)$
queries.  For proving the lower bounds we use graphs generated according to the
\emph{configuration model}~\cite{Bol:1980}. A vector $\dvec = (d_1,d_2, \dots
,d_n)$ is said to be \emph{graphical} if there exists an undirected graph on
$n$ nodes such that vertex $i \in [n]$ has degree $d_i$. We briefly describe
here how graphs are generated in the configuration model: 
\begin{enumerate}
	\item Create disjoint sets $W_i$, for $i \in \{1, \ldots n\}$, with $|W_i| = d_i$.
		The elements of $W_i$ are called \emph{stubs}.
	\item Create a uniform random maximum matching in the set
		$\displaystyle\bigcup_{i=1}^n W_i$  (note that $\sum_{i=1}^n d_i$ must be even
		since $\vd$ is graphical). 
	\item For a stub edge $\{x, y\}$ in the matching, such that $x \in W_i$ and $y
		\in W_j$, the edge $\{i, j\}$ is added to the graph. 
\end{enumerate}
The above procedure creates a graph where vertex $i$ has degree exactly $d_i$.
However, the graph may not be simple, \ie it may have multiple edges and
self-loops. Also, this procedure does not necessarily produce a uniform
distribution over graphs having degree sequence $\vd$.  The expected number of
multi-edges and self-loops in the graph is for many interesting graphs only a
small fraction and  in any graph with bounded degree their expected number is a
constant. 
 We use $G \sim
\mathcal{G}(\vd)$ to denote that a graph $G$ was generated in the configuration
model with degree sequence $\vd$.

Recall the definition of a \emph{sensible} algorithm as one that never makes a
query to which it already knows the answer. A \emph{sensible} algorithm has the
following behaviour: (i) for every $u\in V$ and $i\leq \deg(u)$ ($\deg^+(u)$
and $\deg^-(u)$, respectively) it makes the query $(\ell(u),i)$
($(\ell(u),i,\ein)$ and $(\ell(u),i,\eout)$, respectively) at most once, and
(ii) it never queries a $(\ell(v),i)$ if it has not received $\ell(v)$ as a
valid label or if $i \not \in [1,\deg(v)]$. Note that there exists for any
algorithm a sensible implementation which needs at most as many queries as the
original algorithm.  For technical reasons, in the proof of
\autoref{lem:gato}, we use an oracle $\mathcal{O}^s$ with {\emph
side-information} as an extension of $\mathcal{O}$: $\mathcal{O}^s$ returns,
upon query, exactly the same information as $\mathcal{O}$, but can add
additional truthful information. In particular, we allow the oracle when
queried $(\ell(v),i)$ to not only return the corresponding node $(\ell(u),
\deg(u))$, where $u$ is the $i\th$ element in the adjacency list of $v$, but
also the index, say $j$, in the adjacency list of $u$ which corresponds to $v$.
Clearly, any sensible algorithm wouldn't query $(\ell(u),j)$ after querying
$(\ell(v),i)$.

\begin{lemma}\label{lem:gato}
	Let $\epsilon > 0$. Let $\vd = (d_1, \ldots, d_n)$ be an arbitrary graphical
	sequence and let $D := \sum_{i = 1}^n d_i$. Let $G \sim \mathcal{G}(d)$, and
	let $\tilde{G}$ be an arbitrary graph with degree sequence $\vd$. There
	exists an implementation of an oracle $\uOracle^s(G)$ (with
	side-information) such that if $(\ell(X_t), \deg_G(X_t))_{t=1}^T$ is the
	sequence of responses to a \emph{sensible} algorithm, where
	\[ T := \min\left\{\min \{ \tau \geq 0 ~|~ \text{all neighbours of all known
	nodes are disclosed using $\tau$ queries} \}, \frac{\epsilon}{16} \cdot
	\sqrt{D} \right\}, \]
	and if $(\ell(\tilde{X}_t), \deg_{\tilde{G}}(\tilde{X}_t))_{t=1}^T$ is the
	sequence returned by oracle $\uOracle^{\pi}(\tilde{G})$, then there exists a
	coupling so that the sequences $((\ell(X_t), \deg_G(X_t)))_{t=1}^T$ and
	$((\ell(\tilde{X}_t), \deg_{\tilde{G}}(\tilde{X}_t)))_{t=1}^T$ are identical
	with probability at least $1 - \epsilon$.

\end{lemma}

\begin{remark} 
	Graphs $G \sim \mathcal{G}$ are not necessarily connected. Thus,
	\autoref{lem:gato} uses stopping time $T$ where all edges in the connected
	component involving the starting node, \ie the node sent in response to the
	$\initquery$ are uncovered by the algorithm, either through query responses
	or through side information. Clearly, at this point there is nothing left of
	the algorithm to do except return the number of observed nodes as its
	estimate, which can be arbitrarily far off. 
\end{remark}
\begin{proof}[Proof of \autoref{lem:gato} ]
	We consider the following implementation of $\uOracle^s(G)$, which generates
	the graph $G \sim \mathcal{G}(\vd)$ on the fly as it receives the queries.
	We assume that $\uOracle^s(G)$ uses a consecutive labelling function, \ie
	whenever it chooses a new node, it picks the smallest as yet unused natural
	number as the label. In response to the $\initquery$, it picks $X_1 \sim
	\pi$, sets $\ell(X_1) = 1$ and returns $(\ell(X_1), \deg_{G}(X_1))$.

	It starts by creating sets $W_u = \{ u_1, u_2, \ldots, u_{d_u} \}$ for $u =
	1, \ldots, n$. Let $\mathcal{F}_t$ be the filtration corresponding to all
	the random choices made including the answer to  the $t\th$ query (counting
	the $\initquery$ as the first one). Let $(l, i)$ be the $(t + 1)\th$ query
	made by the algorithm, let $v \in V$, such that $\ell(v) = l$. Since the
	algorithm is \emph{sensible}, it must be that  $(l, i)$ is queried for the
	first time and that the algorithm never learned, through side-information of
	the oracle, to which node the stub $v_i$ connects; this means
	$\uOracle^s(G)$ has not fixed the choice of partner for $v_i$ yet. 	
	
	The oracle chooses a random unmatched stub, say
	$v^\prime_{i^\prime}$ and returns $(\ell(v^\prime), \deg_{G}(v^\prime))$ as
	well as the \emph{side information} $i^\prime$, indicating that $v$ is the
	$(i^\prime)\th$ neighbour of $v^\prime$.

	Let $p^t$ denote the distribution over choice of $v^\prime$ made by
	$\uOracle^s(G)$ before responding to the $(t + 1)\th$ query. Let
	$\delta_u(t)$ denote number of edges of $u$ that have been already disclosed
	up to the answering of the $t\th$ query. This also includes possible
	side-information. Note that $\sum_u \delta(u) = 2 (t - 1)$, since every
	query other than the $\initquery$ reveals information about the two
	end-points of the disclosed edge. Thus, $p^t(u) = (d_u - \delta_u(t))/(D -
	2(t - 1))$. Consider the variation distance between $p^t$ and $\pi$,
	\begin{align*}
		\tvdist{\pi - p^t} &= \frac{1}{2} \sum_{u \in V} \left|\frac{d_u}{D} -
		\frac{d_u - \delta_u(t)}{D - 2(t - 1)}\right| 
		= \frac{1}{2} \sum_{u \in V} \left|\frac{-2(t - 1) d_u  + D
		\delta_u(t)}{D(D - 2t + 2)} \right| \\
		& \leq \frac{1}{2} \sum_{u \in v} \frac{2 (t - 1) d_u}{D (D - 2t + 2)} +
		\frac{1}{2} \sum_{u \in V} \frac{D \delta_u(t)}{D (D - 2t + 2)} \\
		&\leq \frac{1}{2} \cdot \frac{2(t -1 )}{D - 2t + 2} + \frac{1}{2} \cdot
		\frac{2(t - 1)}{D - 2t + 2} = \frac{2t - 2}{D - 2t + 2}
	\end{align*}
	In the final step above, we used the fact that $\sum_u d_u = D$ and $\sum_u
	\delta_u(t) = 2t - 2$. Note that if the stopping condition in the definition
	of $T$ has not occurred the above calculations are valid. Also, we know that
	$T \leq \frac{\epsilon}{16} \sqrt{D}$, thus summing up the variation
	distance $\tvdist{\pi - p^t}$ for $t \in \{1, \ldots ,T\}$ and observing that
	$\tvdist{p^0 - \pi} = 0$, since the response to the $\initquery$ is chosen
	from the stationary distribution, we get the desired result.	\end{proof}

\begin{theorem} \label{thm:lboqmodel}
	Let $\vd = (d_1, \ldots, d_n)$ be a graphical vector satisfying $\min_i d_i
	\geq 3$. Then there exists a graphical $\tilde{\vd} = (\tilde{d}_1, \ldots,
	\tilde{d}_n , \tilde{d}_{n + 1}, \ldots, \tilde{d}_{2n})$ with $\min_i
	\tilde{d}_i \geq 3$, such that for $G \sim \mathcal{G}(\vd)$ and $\tilde{G}
	\sim \mathcal{G}(\tilde{\vd})$, there exists $c > 0$, such that any
	algorithm with access to one of two oracles $\uOracle(G)$ or
	$\uOracle(\tilde{G})$ chosen equal probability, cannot distinguish between
	the two with probability greater than $1 - c$ unless it makes
	$\Omega\left(\frac{1}{\twonorm{\pi_G}} + \davg(G) \right)$ queries to the
	oracle.
\end{theorem}
\begin{proof}
	First we notice that $\frac{1}{\twonorm{\pi_G}} =O(\sqrt{D})$ and $\davg(G)
	= O(\sqrt{D})$ if $D := \sum_{i = 1}^n d_i$. 

	We first prove the lower bound of $\Omega(1/\twonorm{\pi_G})$. We set
	$\tilde{\vd} = (d_1, \ldots, d_n, d_1, \ldots, d_n)$ as in the proof of
	\autoref{lem:perro}. Let $G_1$ and $G_2$ be an arbitrary graphs on $n$ and
	$2n$ vertices with degree sequences $\vd$ and $\tilde{\vd}$ respectively.
	Let $\mathcal{E}_1$ (resp. $\mathcal{E}_2$) be the event that the coupling
	of the sequences output by $\uOracle^s(G)$ (resp. $\uOracle^s(\tilde{G})$)
	and $\uOracle^\pi(G_1)$ (reps. $\uOracle^\pi(G_2)$) holds for the first $T$
	queries using \autoref{lem:gato}. Let $\mathcal{E}^\prime_1$ and
	$\mathcal{E}^\prime_2$ be the events that the stopping condition in the
	coupling in \autoref{lem:gato} happens due to $\epsilon\sqrt{D}/16$
	queries being made. Note that because of the condition $\min_i d_i \geq 3$
	(and $\min_i \tilde{d}_i \geq 3$) the graph $G \sim \mathcal{G}(\vd)$ (and
	$\tilde{G} \sim \mathcal{G}(\tilde{\vd})$) is connected with high probability.
	Thus, as $n \rightarrow \infty$, $\Pr{\mathcal{E}^\prime_1} \rightarrow 1$
	(the same for $\Pr{\mathcal{E}^\prime_2}$). Finally, let $\mathcal{E}_3$ be
	the event that the coupling in \autoref{lem:perro} holds. By choosing
	$\epsilon$ in \autoref{lem:gato} appropriately, we can bound the
	probability that at least one of the events $\mathcal{E}_1$,
	$\mathcal{E}_2$, $\mathcal{E}^\prime_1$, $\mathcal{E}^\prime_2$, and
	$\mathcal{E}_3$ does not occur, by some constant $c < 1$. Thus, we get the
	require result.

	The lower bound of $\Omega(\davg(G))$ follows similarly using
	\autoref{lem:perra} instead of \autoref{lem:perro}. 
\end{proof}

\section{Directed Graphs}\label{sec:directed}

In this section, we consider the query complexity of estimating the number of
nodes in directed graphs.  We first observe that estimating $n$ using the
approach of Katzir \etal~\cite{KLSC14} is not possible since the stationary
distribution of a node is in general not proportional to its degree. Another
obstacle is that the stationary distribution of a node can be exponentially
small as the graphs in Figures~\ref{fig:line} and~\ref{fig:shooting}
illustrate.  In particular, it takes an exponentially large sample drawn from
the stationary distribution to distinguish between the line graph of
Figure~\ref{fig:line} on $n$ nodes and the line graph on $2n$ nodes, since the
probability mass of the additional nodes is $2^{-\Omega(n)}$.  It is also not
very difficult to show that even with access to one of the two oracles
$\dOracle$ or $\dOracleOne$, $\Omega(n)$ queries are required to distinguish
the line graph on $n$ vertices from the line graph on $2n$ vertices.

As the example of the line graph reveals, unlike in the undirected case, rapid
mixing and low average degree are not sufficient conditions to design a good
estimator of the number of nodes using sublinear number of queries. The line
graph shows that in the directed case, rapid mixing does not imply short
(directed) diameter. One might hope that if one throws small diameter into the
mix, in addition to low average degree and rapid mixing, a better estimator
could be designed. In \autoref{sec:comet}, we show that this is not the
case. The problem of estimation remains stubbornly hard, and $\Omega(n)$
queries to the oracle $\dOracle$, and $2^{\Omega(n)}$ queries to the stationary
query oracle are required to achieve a good estimate of the number of nodes.

These observations suggest that exploring the graph, \eg through breadth-first
search, is much faster than sampling from the stationary distribution. The
question of interest is whether there is a property, satisfied by graphs of
interest, which yields a query complexity better than $\Omega(m)$.  We answer
this positively in \autoref{sec:phi}, where we introduce a parameter that
generalises the conductance $\phi_\varepsilon$ and give almost tight bounds on
the number of queries required to estimate $n$ up to an $\varepsilon$ relative
error. Our Algorithm $\mathsf{EdgeSampling}$ takes this parameter as an input
and terminates after $O(n/\phi_\varepsilon)$ queries which can be much smaller
than the  sample complexity of breadth-first search.

Before delving into these results, it is worth pointing out that results in
\autoref{sec:undirected} can be used if access to $\dOracleTwo$ is
provided. In this case, we can simply treat the graph as being undirected.
However, it is still interesting to understand whether the distinction between
in-neighbours and out-neighbours allows one to design better estimators. At
present, we are unaware of any graphs where this might be the case.

\paragraph{The Comet Graph}\label{sec:comet}

\begin{figure}[H]
\begin{center}
	\includegraphics[scale=0.7]{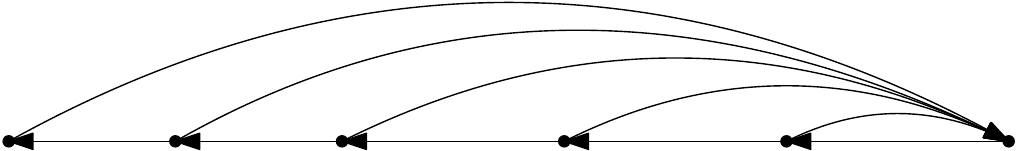}
	\caption{ The \text{Line graph} on $6$ nodes. }\label{fig:line}
\end{center}
\end{figure}

\begin{figure}[H]
\begin{center}
	\includegraphics[scale=0.7]{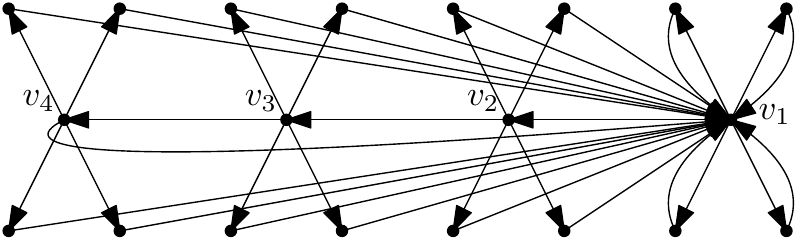}
\caption{ The graph $\text{Comet}(20,4)$.
}\label{fig:shooting}
\end{center}
\end{figure}

\begin{figure}[H]
\begin{center}
	\includegraphics[scale=0.7]{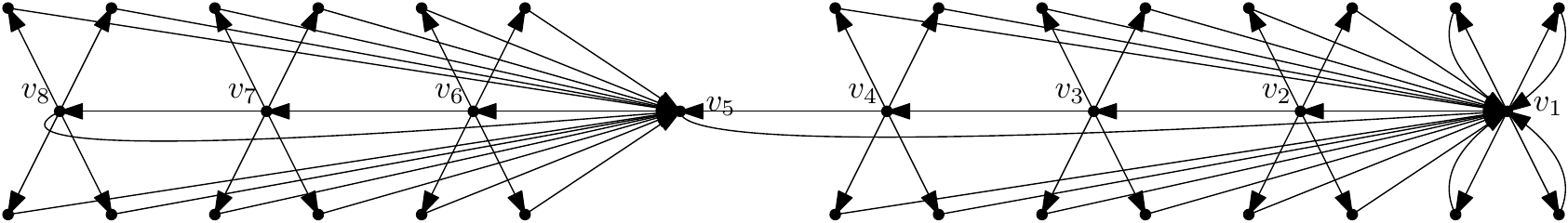}
\caption{ The graph $\text{DoubleComet}(40,8)$.
$\text{DoubleComet}(2n,2k)$ consists of two copies of $\text{Comet}(n,k)$ 
connected in the following way. First we remove all directed neighbours of $v_{k+1}$ as well as   the edge $(v_k,v_1)$. Then we add $(v_k,v_{k+1})$, $(v_{k+1},v_{k+2})$ as well as $(v_{k+1},v_{1})$.
}\label{fig:shooting2}
\end{center}
\end{figure}

The {\it Comet graph}, $\text{Comet}(n,k)$ is constructed as follows. Assume
that $k$ divides $n$. There is a directed cycle on the vertices $v_1,v_2,\dots,
v_k$, with edges $(v_i,v_{i+1})$  for $1\leq i <k$ and $(v_k,v_1)$.  We denote
these $k$ vertices as \emph{centres}.  For every $\ell \in [k]$, there is a
directed star $S_{\ell}=\{(v_{\ell},v_{\ell, j}):j\in [n/k-1]\}$ with centre in
$v_{\ell}$ of degree $n/k-1$. For each leaf $v_{\ell, j}$ in star $S_\ell$ with
$\ell\in [k],j \in[n/k-1]$, there is a directed edge to the first star centre
$v_1$, that is, $\{(v_{\ell, j},v_{1}):\ell \in [1,k], j\in [n/k-1]\}$.  We
write  $v_\ell^{ G}$ and $v_{\ell,j}^G$ to emphasise that the nodes belong to
graph $G$.

In the following lemma we obtain bounds on stationary distribution of the
nodes, which will allow us to obtain a bound on the query complexity w.r.t.
$\mathcal{O}^\pi$ (see \autoref{thm:Comet}).
\begin{lemma}\label{ShootTheDistributionIsTerribel}
	Let $n$ be a multiple of $k$. Let $\pi$ be the
	stationary distribution of the nodes in  $\text{Comet}(n,k)$.  Then, for $G=\text{Comet}(n,k)$ and $\tilde G
	    =\text{Comet}(2n,2k)$ it holds that $\tvdist{\pi_G - \pi_{\tilde G}}=O((\frac{k}{n})^{k-1}) $.
\end{lemma}
\begin{proof}
    Note that for $u\in V$ we have $\pi(v)=\sum_{u \in V} p(u,v)\pi(u)$. 
     Let $\ell \in [k]$, $u\in N^+(v_{\ell})\setminus \{ v_1 \}$ and $d=n/k$. We have that 
\begin{align}\label{agua}	 
\pi(u)=\pi(v_{\ell})/2d+\pi(u)/2.
\end{align}
 This implies that $\pi(v_{\ell})=d\pi(u)$ for every $u\in N^+(v_{\ell})$. In particular, for $u=v_{\ell+1}$ (consider $v_{k+1}=v_1$) it follows that 
\begin{equation}\pi(v_{\ell+1})=\pi(v_{\ell})/d.\end{equation}
%
 Therefore, $\pi(v_{\ell})=\pi(v_1)/d^{\ell-1}$. Furthermore, 
 \begin{align}\label{fuego}
 \pi(v_{\ell, j})=\pi(v_{\ell+1})=\pi(v_1)/d^{\ell}
 \end{align}
  for every $\ell \in [k]$, $j\in [d-1]$.	
 Summing up everything,
\begin{eqnarray*}
1&=&\sum_{\ell=1}^{k}\left(\pi(v_{\ell})+\sum_{j=1}^{d-1}\pi(v_{\ell, j})\right)\\
 &=&\sum_{\ell=1}^{k}\frac{\pi(v_1)}{d^{\ell-1}}+\sum_{\ell=1}^{k}\sum_{j=1}^{d-1}\frac{\pi(v_1)}{d^{\ell}}\\
   &=&\pi(v_1)\left(\frac{1-1/d^k}{1-1/d}+ (d-1)\frac{1/d-1/d^{k+1}}{1-1/d}\right),\\
  &=&\pi(v_1)\left( \frac{2d -1 - \frac{2d+1}{d^k}}{d-1}  \right).
\end{eqnarray*}
Rearranging the terms yields 
\begin{align}\label{oso}
\pi(v_1)=\frac{d-1}{2d -1 - \frac{2d+1}{d^k}}.
\end{align}
We proceed by establishing a bound on $	\tvdist{\pi_G - \pi_{\tilde G}}$.
  	\begin{align*}
		\tvdist{\pi_G - \pi_{\tilde G}} &= 
		\frac{1}{2} \sum_{\ell=1}^k \left|  \pi(v_\ell^G) - \pi(v_\ell^{\tilde G})        \right|
		+\frac{1}{2} \sum_{\ell=1}^k\sum_{j=1}^{d-1} \left|  \pi(v_{\ell,j}^G) - \pi(v_{\ell,j}^{\tilde G})        \right|  
		   \\
		   &\phantom{oo}+ \frac{1}{2} \sum_{\ell=k+1}^{2k} \left| 0- \pi(v_\ell^{\tilde G})   \right|   + 
		   \frac{1}{2} \sum_{\ell=k+1}^{2k}\sum_{j=1}^{d-1} \left|  0 - \pi(v_{\ell,j}^{\tilde G})  \right|
		   \\
		   &\leq
		   		 \sum_{\ell=1}^k \left|  \pi(v_\ell^G) - \pi(v_\ell^{\tilde G})        \right|
		+ \sum_{\ell=k+1}^{2k} \left| \pi(v_\ell^{\tilde G})   \right| 
		   \\
		&\leq   \left|  \pi(v_1^G) - \pi(v_1^{\tilde G})        \right|
\sum_{\ell=1}^k \frac{1}{d^{\ell-1}} + \sum_{\ell=k+1}^{2k} \frac{1}{d^{\ell-1}}\\
&\leq 2 \left|  \pi(v_1^G) - \pi(v_1^{\tilde G})        \right| + \frac{\frac{1}{d^k} }{1-\frac{1}{d}}.
	\end{align*}
	 By triangle inequality, 
 \[
  \left|  \pi(v_1^G) - \pi(v_1^{\tilde G})        \right| \leq    \left|  \pi(v_1^G) - \    \frac{d-1}{2d -1} \right|   +    \left|  \pi(v_1^{\tilde G}) -    \frac{d-1}{2d -1} \right| \leq 2  \left|  \frac{d-1}{2d -1 - \frac{2d+1}{d^k}} -    \frac{d-1}{2d -1} \right| \leq 4   \frac{2d+1}{d^k}.  
  \]
Putting everything together yields $	\tvdist{\pi_G - \pi_{\tilde G}} = O\left (\left(\frac{k}{n}\right)^{k-1}\right).$

\end{proof}

To establish the  mixing time of $\text{Comet}(n,k)$, we will make use of the
result of Levin, Peres and Wilmer \cite{LPW06} which relates the mixing time of
a Markov chain to the probability that two copies of the chain meet.

\begin{theorem}\label{thm:Comet}
Let $n$ be a multiple of $k$. Then $\text{Comet}(n,k)$ has   mixing time  $\tmix=O(1)$ and diameter $k$. Furthermore, any algorithm requires at least $\Omega((\frac{n}{k})^{k-1})$ queries to $\mathcal{O}^\pi$ and  $\Omega( n)$  queries to $\dOracle$ to distinguish between     $G=\text{Comet}(n,k)$ and $\tilde G =\text{Comet}(2n,2k)$.	\end{theorem}
\begin{proof}
Observe that the diameter of the graph is indeed $k$. 
Let $(X_t)_{t\geq 0}$, $(Y_t)_{t\geq 0}$ be random walks on $\text{Comet}(n,k)$  with $X_0=x$ and $Y_0=y$, where $x,y \in V$.
Whenever $X_s= Y_s$, then we can couple $X_t=Y_t$ for $t\geq s$.
Let $T= \min \{ t\geq 0 ~|~ X_t=Y_t\}.$
 Starting from any vertex $u \in V$,  with constant probability  at least $p > 0$, the random walk is at $v_1$ after  $2$ steps, \ie $\Pr{X_2=v_1~|~X_0=u} \geq p$ for $u\in V$.
 In particular, using independence, $\Pr{X_2=Y_2~|~X_0=u,Y_0=v} \geq \Pr{X_2=v_1~|~X_0=u}\cdot \Pr{Y_2=v_1~|~Y_0=v} \geq p^2$ for $u,v\in V$.
Iterating this and using independence shows 
  \[
  \Pr{T > 2t} \leq \Pr{X_{2t}\neq Y_{2t}~|~X_0=u,Y_0=v} \leq (1- p^2)^t\leq e^{-1}.
  \]
	For some large enough constant $t$.  From \autoref{lem:LPW06}
	(\cite{LPW06}) and using the definition of the mixing time we derive
	$\tmix=O(1)$.  We now consider the query complexity w.r.t.
	$\mathcal{O}^\pi$.  By \autoref{ShootTheDistributionIsTerribel} we have
	that $	\tvdist{\pi_G - \pi_{\tilde G}} = O\left
	(\left(\frac{k}{n}\right)^{k-1}\right).$ This implies that the first
	$\Omega\left(\left(\frac{n}{k}\right)^{k-1}\right)$ samples can be coupled
	with constant probability.  Thus, in order to distinguish between $G$ on $n$
	nodes and $\tilde G$ on $2n$ nodes one requires
	$\Omega\left(\left(\frac{n}{k}\right)^{k-1}\right)$ samples.  We now
	consider the query complexity w.r.t. $\dOracle$.  In response to the
	$\initquery$ query, both oracles $\dOracle(G)$ and $\dOracle(\tilde G)$ pick
	$v_1$, set $\ell(v_1) = 1$ and return $(\ell(v_1), \deg^+(v_1))$.  Any
	algorithm has to reach at least node $v_k$ in order to distinguish between
	$G$ and $\tilde G$. We assume that the oracles $\dOracle(G)$ and
	$\dOracle(\tilde{G})$ use independent random adjacency lists at all nodes.
	But we couple them to be the same for the nodes that are common to both
	graphs. Clearly, unless $\Omega(n/k)$ queries are made to the out-neighbour
	lists of each of $v_1, v_2, \ldots, v_{k-1}$, the node $v_k$ cannot be
	discovered, which in turn means that it is impossible to distinguish between
	the oracles $\dOracle(G)$ and $\dOracle(\tilde{G})$.  Thus, $\Omega( k \cdot
	n/k) = \Omega(n)$ queries are required.  This concludes the proof.
\end{proof}
Note that the above results only apply to the oracles $\dOracle$ and
$\mathcal{O}^\pi$, but not to $\dOracleOne$, since the in-degrees make it easy
to distinguish between the two graphs.  However, it is straightforward to
extend the graph such that the sample complexity remains $\Omega(n)$ even if
the in-degrees are known; thus even with access to $\dOracleOne$, $\Omega(n)$
queries are required. 

\begin{observation}\label{obs:doublecomet}
	Let $n$ be a multiple of $k$. Then $\text{DoubleComet}(2n,2k)$ (Defined in
	Figure~\ref{fig:shooting2}) has mixing time   $\tmix=O(1)$ and diameter
	$2k$. Furthermore, any algorithm requires at least $\Omega(n)$ queries to
	$\dOracleOne$ to distinguish between 	$G=\text{Comet}(n,k)$ and $\tilde G
	=\text{DoubleComet}(2n,2k)$ on $2n$ nodes.
\end{observation}

The proof is along the same lines as the proof of \autoref{thm:Comet}.
Observe that the in-degrees of nodes $v_1,\dots, v_k$ give no extra information
about the size of the graph.
   
\subsection{Assuming a Bound on the Connectivity}\label{sec:phi}

In this section we introduce the parameter $\emph{general conductance}$. We
first recall some graph notation in the directed setting. Given a non-empty
proper subset of vertices $S\subset V$, let $\deg^+(S)=|\{(u,v)\in E:u\in S\}|$
be the out-degree of $S$. The {\it cut} of $S$, $\partial S$, is the set of
edges crossing between $S$ and $V\setminus S$, that is, $\partial S=\{(u,v)\in
E:u\in S,v\notin S\}$.  The \emph{general conductance} of $S$, $\phi(S)$, is
the ratio between the cut of $S$, and the out-degree of $S$. That is,
$\phi(S)=|\partial S|/\deg^+(S).$ Given $\varepsilon>0$, the graph
$\varepsilon$-general conductance, $\phi_{\varepsilon}$, is the minimum of
$\phi(S)$ over every non-empty proper subset of $V$ of size at most
$(1-\varepsilon)|V|$, \ie 
$
\phi_\varepsilon(G)=\min_{\substack{S\subseteq V:1\le |S|\le (1-\varepsilon)|V|}}\phi(S).
$ 
Note that the parameter $\phi_\varepsilon$ decreases monotonically as
$\varepsilon$ decreases. In the undirected setting for $\varepsilon=1/2$ this
is just what is commonly known as the conductance.\footnote{The definition of
conductance in the directed setting is more involved and more importantly
doesn't seem be directly relevant to the question of estimating the number of
nodes. It suffers from similar problems as the skewed stationary distributions.
Graphs having poor connectivity to a large fraction of the nodes may still have
very \emph{good} conductance if the total mass of the poorly connected nodes
under the stationary distribution is very small.}

In the following we describe the algorithm that estimates the graph size. 

\paragraph{Upper bound in terms of the general conductance.} 

We consider algorithm $\mathsf{EdgeSampling}$  for estimating the number of
nodes. The algorithm takes as input the parameter $\phi$, a lower bound  on the
general conductance $\phi_\varepsilon$.  The query complexity is
$O(n/\phi_\varepsilon)$ and the output estimate $\hat n$ satisfies
$(1-\varepsilon)n \leq \hat n \leq n$ with arbitrary confidence controlled by
an input parameter~$\ell$.  Observe that $O(n/\phi_\varepsilon)$ can be much
smaller than the run time of breadth-first search $\Omega(m)$. \medskip 

\noindent  {\it Algorithm overview.} The algorithm works as follows.  At each
time step the algorithm maintains a counter $Y$.  If at some point the counter
exceeds the threshold $\ell$, then the algorithm terminates.  The algorithm
divides the queries into blocks of length at most $2/\phi$ corresponding to the
execution of the  {\bf for} loop.  In each block, at every step the algorithm
samples one outgoing edge uniformly at random from those available and not
queried before.  If at any step a new node is disclosed, then this finishes the
block (break of the {\bf for} loop) and the counter $Y$ is decreased by $1$.
If the the block finishes without finding a new node, then the counter is
increased by $1$.  Once the counter reaches $\ell$, which will happen
eventually, then the algorithm outputs the number of nodes it discovered.  Even
though our goal is to minimise the query complexity, it is worth noticing that
the time and space complexity can be kept low by choice of suitable data
structures. Although the oracle returns labels of nodes, we use nodes and their
labels interchangeably in the algorithm and the analysis of
\autoref{thm:upper}.
   

\begin{algorithm}
\caption{$\mathsf{EdgeSampling}(\dOracle,\ell,\phi)$}
\label{algo:edge-sampling}
\begin{algorithmic}[1]
\STATE $Y_0=0$ (fail surplus counter)
\STATE $v = (node)~  \dOracle.\initquery$ (query oracle to get the initial node)
\STATE $S_0= \{ v\}$
\STATE $E_{0} = \{ (v,i) ~|~ i\leq \deg^+(v)  \}$ (set of undisclosed edges)
\STATE $t=1$
\WHILE{ $Y_t \leq \ell $ }
\FOR{ $\tau=1$ \TO $2/\phi$ }
\STATE  choose $(u,i)$ uniformly at random from $E_{t,\tau-1}$.
\STATE  $v=  (node)~ \dOracle.(u,i,\eout)$
\IF{$v \not\in S_{t-1}$ }
\STATE $S_t=S_{t-1} \cup \{ v\}$
\STATE $E_{(t-1)\cdot 2/\phi+\tau} \leftarrow (E_{(t-1)\cdot 2/\phi+\tau-1} \cup \{  (v,i) ~|~ i\leq \deg^+(v)  \}) \setminus \{ (u,i) \} $
\BREAK{ }
\ELSE
\STATE $E_{t\cdot 2/\phi} \leftarrow E_{t\cdot 2/\phi-1} \setminus \{ (u,i) \} $
\ENDIF
\ENDFOR
\IF{$|S_t| = |S_{t-1}|+1$}
\STATE $Y_t \leftarrow Y_t -1$
\ELSE
\STATE $Y_t \leftarrow Y_t + 1$
\STATE $S_t \leftarrow S_{t-1}$
\ENDIF
\STATE $t\leftarrow t+1$

\ENDWHILE
\STATE Output $|S_t|$.
\end{algorithmic}
\end{algorithm}


\begin{theorem}\label{thm:upper}
	Algorithm $\mathsf{EdgeSampling}(\dOracle,\ell,\phi)$ on graph $G$ has a query complexity of 
	$\min\{ 2(2n+\ell)/\phi, m\}$
	and outputs an estimate $\hat n \leq n$.
	 Furthermore, if $G$ has general conductance $\phi_{\varepsilon}(G)$ of at least $\phi$, then the algorithm satisfies $\hat n \geq (1-\varepsilon)n$ w.p. at least $1-2^{-\ell}$.
\end{theorem}

\begin{proof}
	First observe, that per iteration of the {\bf for} loop at most one new node
	is discovered.  Let  $Y_t$ be  the variables defined in the algorithm and
	define  $X_t=|S_t| - |S_{t-1}|$. In particular $X_t\in \{0,1\}$.  Note that
	the number of times  the counter $Y$ decreases, \ie $Y_t < Y_{t-1}$,  during
	the first $\tau$ iterations of the {\bf for} loop is $\sum_{i\leq \tau} X_i
	\leq n $, since the counter increases  when a new node is discovered.
	Therefore, the number of queries is bounded by $2\cdot (2n+\ell)/\phi$.  Moreover,
	every edge is only queried at most once yielding the claimed  bound on the
	query complexity.  We now prove that  the output $\hat n := |S_t|$ satisfies
	$\hat n \geq (1-\varepsilon)n$ under the conditions in the statement.

In the remainder we assume $\phi_{\varepsilon}(G) \geq \phi$  since otherwise the statement is trivially true.
Let 
\[T=\min \left\{ t\geq 0 ~\Big|~  \sum_{i=0}^t X_i \geq n(1-\varepsilon) \text{ or }  Y_t =\ell  \right\}.
\]

Note that every iteration  $j\leq T$  of the {\bf for} loop satisfies that the probability of finding a new node is at least $|\partial S_{j}|/\deg^+(S_{j}) \geq \phi_\varepsilon(G)\geq \phi$.
      Hence, we have that $\Pr{X_j=0}\leq (1-\phi)^{2/\phi} < 1/3. $
Observe that $X_t =1$ implies $Y_t=Y_{t-1}-1$ and $X_t=0$ implies $Y_t=Y_{t-1}+1$.
Consider the Markov chain $(Z_t)_{t\geq 0}$ defined in \autoref{pro:CaminataAleatoriaParcial} with $s=n,b=n+\ell$ and $p=1/3$.
We  couple $Y_t$ and $Z_t$ for $t\leq T$ such that $Y_t  + n \leq Z_t$, since $\Pr{Y_t=Y_{t-1}+1}=\Pr{X_t=0} \leq 1/3=p=\Pr{Z_t=Z_{t-1}+1}$. 
Let 
\[T'=\min \left\{ t\geq 0 ~\Big|~  Z_t \in \{ 0, n+\ell \}  \right\}.
\]
From \autoref{pro:CaminataAleatoriaParcial} we get that
\[
\Pr{Z_{T'}=n+\ell}=
\frac{\left(\frac{2/3}{1/3}\right)^{n} - 1  }{ \left(\frac{2/3}{1/3}\right)^{n+\ell} - 1}\leq 2^{-\ell}.
\]
Due to our coupling we have $Y_t+n \leq Z_t$ and observe that $Z_t=0$ implies that $\sum_{i}^t X_i \geq n$  and therefore we have $T\leq T'$.
Hence,
$Z_{T'}=0$ implies $Z_{T} < n+\ell$ which in turn implies 
$Y_T < \ell$.
Thus,
\begin{eqnarray*}
\Pr{|S_T|\geq n(1-\varepsilon) } &=& \Pr{\sum_{i=0}^T X_i \geq n(1-\varepsilon) }\\ 
						&=& \Pr{Y_T< \ell}\\
						&\geq&  \Pr{Z_T< n+\ell}\\
						&\geq&  \Pr{Z_{T'}< n+\ell} \\
						&=&  1-\Pr{Z_{T'}=n+\ell}\\
						&\geq&  1-2^{-\ell}, 	
\end{eqnarray*}
 %
which concludes the proof.
\end{proof}

Observe that the error made by Algorithm  $\mathsf{EdgeSampling}$ is one-sided --- the estimate never exceeds $n$.
Allowing a two-sided error and  given knowledge of $\varepsilon$, one can instead output an estimate, which has a smaller additive error.
This is summarised in the following observation. 
\begin{observation}\label{obs:upper}
	Consider a modification of Algorithm
	$\mathsf{EdgeSampling}(\dOracle,\ell,\phi)$ on graph $G$ which takes the
	additional parameter $\varepsilon$  and  outputs $ \hat n^* :=
	|S_t|(1+\frac{\varepsilon}{2-\varepsilon})$ instead of $|S_t|$.  If $G$ has
	general conductance $\phi_{\varepsilon}(G)$ of at least $\phi$, then $\hat
	n^*$ satisfies $|n - \hat n^*|  \leq \frac{\varepsilon}{2-\varepsilon}n$
	w.p. at least $1-2^{-\ell}$.
\end{observation}

\begin{proof}
By \autoref{thm:upper} we have  $n\geq |S_t| \geq (1-\varepsilon)n$, where $t$ is the the index of the last execution of the ${\bf while}$ loop. 
We have that 
\begin{eqnarray*} 
\left|n -  \hat n^*\right|   &=& \left|n -  |S_t|\left(1+\frac{\varepsilon}{2-\varepsilon}\right)\right|  \\
				&\leq& \max \left\{  \left|n -  n\left(1+\frac{\varepsilon}{2-\varepsilon}\right)  \right|,  \left|n -  n(1-\varepsilon)\left(1+\frac{\varepsilon}{2-\varepsilon}\right) \right|   \right\}\\
				&=&  \frac{\varepsilon}{2-\varepsilon}n,
\end{eqnarray*}
which finishes the proof.								
\end{proof}

\paragraph{Lower bound in terms of the general conductance.} 

In the following we show that the bound of \autoref{obs:upper} is
almost tight.  Recall that, given $\phi_\varepsilon$, the modified version of
Algorithm $\mathsf{EdgeSampling}$ in \autoref{obs:upper} returns an
estimate with and additive error of at most
$\frac{\varepsilon}{2-\varepsilon}n$ using $O(n/\phi_\varepsilon)$ queries.  In
what follows we show that any algorithm, given the values $\phi_\varepsilon$
and $\varepsilon$, cannot output an estimate with an error smaller than
$\frac{\varepsilon-\delta}{2-\varepsilon-\delta}n$ unless it makes
$\Omega(n/\phi_\varepsilon)$ queries, for any $\delta<\varepsilon/2$. We prove
the following lemma for undirected graphs using $\uOracle$, but it should be
clear that the same the same result holds for directed graphs, by making the
graph directed, with symmetric edges, and using $\dOracleTwo$ (and hence also
for oracles $\dOracle$ and $\dOracleOne$). 

\begin{figure}[H]
\begin{center}
		\includegraphics[scale=0.7]{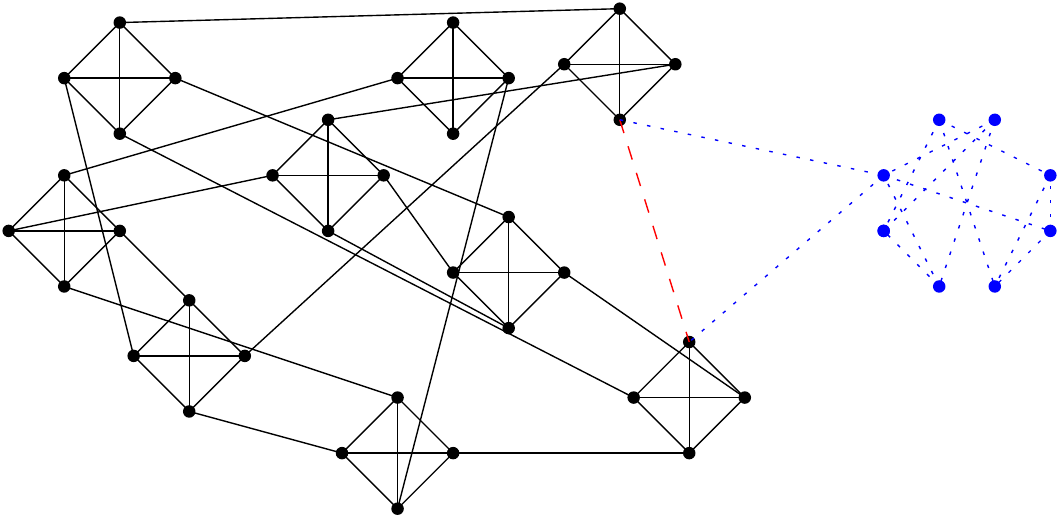}

\end{center}
\caption{
The graphs of \autoref{thm:philower}.
	$G$ contains the black nodes and the black and red (dashed) edges which form a $\ceil{1/\phi}$-regular expander on cliques of size $\ceil{1/\phi}$.
	$G'$  is obtained by removing the red (dashed) edge and adding the blue (dotted) graph.
At least one blue edge needs to be sampled, which takes $\Omega(n/\phi)$ time,  in order to estimate $n$ accurately. 
 }
 \label{fig:HideIt}
\end{figure}

\begin{theorem}\label{thm:philower}
	Let $n\in \mathbb{N}$, $\phi\in [1/n, 1]$ and $\varepsilon \in (0, 1/2]$.
	There exists an undirected graph with general conductance
	$\phi_{\varepsilon}=\Theta(\phi)$	such that any algorithm with access to
	$\uOracle$ requires $\Omega( n/\phi_{\varepsilon})$ queries to output $\hat
	n$ such that $|n-\hat n| \leq
	\frac{\varepsilon-\delta}{2-\varepsilon-\delta} n$ w.p. at least $2/3$ for
	any $\delta<\varepsilon/2$. 
\end{theorem}

\begin{proof}
	Let $d=\ceil{1/\phi}$.    For simplicity we assume that $d$ divides $n$.
	Let $G_1=(V_1,E_1)$ be a $d$-regular expander on $n^*=n(1-\varepsilon+\delta/2)/d$ nodes.
	We construct $G=(V,E)$ (see Figure~\ref{fig:HideIt} for an illustration) by replacing every node $r$ by a clique $K$ of size $d$ such that every node of the clique has exactly one of the edges incident to $r$. Let $n_1 := |V|=n(1-\varepsilon+\delta/2)$.	
%
	Let $G_2$ be a $3$-regular expander on $n-n_1$ nodes.
	 Choose an edge $(u,v)$ uniformly at random among all edges where both nodes $u$ and $v$ belong to the same clique.
	 We construct $G'$ by removing $(u,v)$, and by adding the edge $(u,w),(v,w)$, where  $w$ is an arbitrary node of $G_2$.
	 
	 The proof idea is to couple the oracle decisions such that
	 w.p. $2/3$ no algorithm can distinguish between $G$ and $G'$. Incidentally,  our construction ensures that 	  there is no estimate $\hat n$ satisfying $|\hat n - n_1| \leq\frac{\varepsilon-\delta}{2-\varepsilon-\delta}n_1$ and $|\hat n - n| \leq \frac{\varepsilon-\delta}{2-\varepsilon-\delta} n$ simultaneously. 
	 Suppose for the sake of contradiction this was possible, then 
	 $\hat n - n(1-\varepsilon+\delta/2)\leq \frac{\varepsilon-\delta}{2-\varepsilon-\delta}n_1  $
	 and $n - \hat n \leq \frac{\varepsilon-\delta}{2-\varepsilon-\delta}   $.
	 Adding up both equations yields
	 $n-n(1-\varepsilon+\delta/2)\leq \frac{\varepsilon-\delta}{2-\varepsilon-\delta}(n_1+n)$. This 
	 implies $\varepsilon-\delta/2 \leq \frac{\varepsilon-\delta}{2-\varepsilon-\delta} (2-\varepsilon+\delta/2)$,	  which can be shown to a contradiction for $\varepsilon \leq 1/2$. 
	 Thus, any algorithm approximating $n$ up to an error
	 of $\frac{\varepsilon-\delta}{2-\varepsilon-\delta}$
	   needs to distinguish between $G$ and $G'$.

	 	We assume that the oracles $\uOracle(G)$ and $\uOracle(G')$ use a consecutive labelling function, \ie
	whenever it chooses a new node, it picks the smallest as yet unused natural
	number as the label. In response to the $\initquery$ query, both oracles pick $v_1 \in V_1\setminus \{u,v \}$, set $\ell(v_1) = 1$ and return $(\ell(v_1), \deg_{G}(v_1))$.
	As long as the algorithm doesn't disclose the edge $(u,w)$ or $(v,w)$ in $G'$ and $(u,v)$ in $G$, 
	we can couple the nodes returned by the oracle and thus
	the two graphs are indistinguishable.
	Hence, similarly as before,  we couple with constant probability for the first $cn/\phi$ samples returned by the oracle for constant $c>0$ small enough.  
     It remains to show that  $G'$ has general conductance $\phi_{\varepsilon}=\Theta(\phi)$.
     
     Consider an arbitrary clique $K$ in the construction. Clearly, $\phi_{\varepsilon} \leq \phi(V(K))=O(d/d^2)=O(\phi)$. To derive a lower bound on $\phi_{\varepsilon}$ consider any set $S$ such that $\phi=\phi(S)$.
Suppose that $S$ contains only entire cliques, \ie if $u\in S$, then also $v\in S$ for $u,v$ belonging to the same clique.
Then, the general conductance of $S$ is $\phi(S)=\Theta(\phi^{G_1}/d)=\Theta(\phi)$, where $1/\phi^{G_1}= \Theta(1)$ is the general conductance of the $G_1$.

We claim that there exists a set of cliques, $\mathcal{K}$, containing a  constant fraction of the cliques and no node of $S$ is in any clique of $\mathcal{K}$. If this were not the case, then $S$ either violates $|S|\leq n(1-\varepsilon)$ or  $S$ violates $\phi=\phi(S)$. 

Now suppose that $S$ does not contain only entire cliques. Observe that for any clique $K$ which is not entirely included in $S$, we can include the remaining nodes which only decreases the general conductance.
The new resulting set $S'$ might be larger than $n(1-\varepsilon)$. However, observe that $\phi_{\varepsilon}\geq \phi_{\varepsilon-\varepsilon'}$ for any sufficiently small $\delta>\varepsilon'>0$.
Note that $S'$ does not contain any nodes of cliques in $\mathcal{K}$. Hence, the resulting general conductance fulfils $ \phi(S) \geq \phi(S') \geq \Omega(\phi^{G_1}/d)=\Omega(\phi)$.
 This completes the proof.								\end{proof}


%
\bibliographystyle{plain}
\bibliography{biblio}

\begin{thebibliography}{10}

\bibitem{BR:2002}
Michael~A. Bender and Data Ron.
\newblock Testing properties of directed graphs: Acyclicity and connectivity.
\newblock {\em Random Structures and Algorithms}, 20(2):184--205, 2002.

\bibitem{Bol:1980}
Béla Bollobás.
\newblock A probabilistic proof of an asymptotic formula for the number of
  labelled regular graphs.
\newblock {\em European Journal of Combinatorics}, 1(4):311 -- 316, 1980.

\bibitem{CF02}
Colin Cooper and Alan Frieze.
\newblock Crawling on web graphs.
\newblock In {\em Proceedings of the thiry-fourth annual ACM symposium on
  Theory of computing}, pages 419--427. ACM, 2002.

\bibitem{CRS12}
Colin Cooper, Tomasz Radzik, and Yiannis Siantos.
\newblock Estimating network parameters using random walks.
\newblock In {\em Computational Aspects of Social Networks, 4th International
  Conference on}, pages 33--40, 2012.

\bibitem{CPS:2016}
Artur Czumaj, Pan Peng, and Christian Sohler.
\newblock Relating two property testing models for bounded degree directed
  graphs.
\newblock In {\em Proceedings of the ACM Symposium on the Theory of Computing
  (STOC)}, 2016.

\bibitem{DKS14}
Anirban Dasgupta, Ravi Kumar, and Tam{\'{a}}s Sarl{\'{o}}s.
\newblock On estimating the average degree.
\newblock In {\em Proceedings of the 23rd international conference on World
  Wide Web}, pages 795--806. ACM, 2014.

\bibitem{f68}
William Feller.
\newblock {\em An Introduction to Probability Theory and Its Applications},
  volume~1.
\newblock Wiley, 1968.

\bibitem{Finkel98}
Mark Finkelstein, Howard~G Tucker, and Jerry~Alan Veeh.
\newblock Confidence intervals for the number of unseen types.
\newblock {\em Statistics \& Probability Letters}, 37(4):423--430, 1998.

\bibitem{Gol:2010}
Oded Goldreich.
\newblock Introduction to testing graph properties.
\newblock Survey article available at
  \url{http://www.wisdom.weizmann.ac.il/~oded/COL/tgp-intro.pdf}, 2010.

\bibitem{Goo:1953}
I.~J. Good.
\newblock The population frequencies of species and the estimation of
  population parameters.
\newblock {\em Biometrika}, 40(3/4):237--264, 1953.

\bibitem{HK:13}
Liran Katzir and Stephen~J. Hardiman.
\newblock Estimating clustering coefficients and size of social networks via
  random walk.
\newblock {\em ACM Transactions on the Web}, 9(4):19:1--19:20, 2015.

\bibitem{KLSC14}
Liran Katzir, Edo Liberty, Oren Somekh, and Ioana~A. Cosma.
\newblock Estimating sizes of social networks via biased sampling.
\newblock {\em Internet Mathematics}, 10(3-4):335--359, 2014.

\bibitem{knuth74}
Donald~E. Knuth.
\newblock Estimating the efficiency of backtrack programs.
\newblock {\em Mathematics of computation}, 29(129):122--136, 1975.

\bibitem{LPW06}
David Levin, Yuval Peres, and Elizabeth Wilmer.
\newblock {\em {Markov chains and mixing times}}.
\newblock American Mathematical Society, 2006.

\bibitem{marchetti88}
Alberto Marchetti-Spaccamela.
\newblock On the estimate of the size of a directed graph.
\newblock In {\em International Workshop on Graph-Theoretic Concepts in
  Computer Science}, pages 317--326. Springer, 1988.

\bibitem{MST06}
Milena Mihail, Amin Saberi, and Prasad Tetali.
\newblock Random walks with lookahead on power law random graphs.
\newblock {\em Internet Mathematics}, 3(2):147--152, 2006.

\bibitem{MSL16}
Cameron Musco, Hsin{-}Hao Su, and Nancy~A. Lynch.
\newblock Ant-inspired density estimation via random walks: Extended abstract.
\newblock In {\em Proceedings of the 2016 {ACM} Symposium on Principles of
  Distributed Computing, {PODC} 2016, Chicago, IL, USA, July 25-28, 2016},
  pages 469--478, 2016.

\bibitem{VV11}
Gregory Valiant and Paul Valiant.
\newblock Estimating the unseen: an n/log(n)-sample estimator for entropy and
  support size, shown optimal via new clts.
\newblock In {\em Proceedings of the forty-third annual ACM symposium on Theory
  of computing}, pages 685--694. ACM, 2011.

\bibitem{VV13}
Gregory Valiant and Paul Valiant.
\newblock Estimating the unseen: improved estimators for entropy and other
  properties.
\newblock In {\em Advances in Neural Information Processing Systems}, pages
  2157--2165, 2013.

\end{thebibliography}
\appendix

\section{Auxiliary Claims}
\begin{theorem}[\cite{LPW06}]\label{lem:LPW06} 
Consider two irreducible Markov chains $(X_t)_{t\geq 0}$, $(Y_t)_{t\geq 0}$ with transition matrix $P$, $X_0=x$, and $Y_0=y$.
	Let $\{(X_t, Y_t) \}$ be a coupling satisfying that if
	$X_s= Y_s	$, then $X_t=Y_t$ for $t\geq s$.
	Let $T= \min \{ t\geq 0 ~|~ X_t=Y_t\}.$ 
	Then $\tvdist{p^t(x,\cdot)-p^t(y,\cdot)}\leq \Pr{T > t}$.
\end{theorem}

\begin{proposition}{\cite[Chapter XIV.2]{f68}}]\label{pro:CaminataAleatoriaParcial}
Let $p\in (0,1/2)$ and $b,s\in\naturals$. Consider a discrete time Markov chain $(Z_t)_{t\geq 0}$
 with state space $\Omega=[0,b]$ where
\begin{itemize}
	\item $Z_0=s \in [0,b]$
 	\item  $\Pr{Z_t=i ~|~ Z_{t-1}=i-1}=p$ for $i\in [1,b-1] , t\geq 1$
	\item $\Pr{Z_t=i ~|~ Z_{t-1}=i+1}=1-p$ for $i\in [1,b-1] , t\geq 1$
	\item $\Pr{Z_t=i ~|~ Z_{t-1}=i}=1$ for $i\in \{0,b \}, t\geq 1$ 
\end{itemize}
	Let $T= \min \{ t\geq 0 ~|~ Z_t \in \{ 0,b \}\}$. Then
	\[ 
	\Pr{Z_T=b}=\frac{\left(\frac{1-p}{p}\right)^{s} - 1  }{ \left(\frac{1-p}{p}\right)^{b} - 1}.
	\] 
\end{proposition}
%












\end{document}